\documentclass[12pt]{article}

\usepackage[english]{babel}

\usepackage[margin=1in]{geometry}  
\usepackage{setspace}              
\usepackage{xcolor}                 
\usepackage{times}                 
\usepackage[authoryear]{natbib}

\usepackage{graphicx}
\usepackage[colorlinks=true, allcolors=blue]{hyperref}
\usepackage{float}
\usepackage{verbatim}
\usepackage{multirow}
\usepackage{graphicx} 
\usepackage{booktabs} 
\usepackage{amsmath,amssymb, amsbsy}
\usepackage{mathtools}
\usepackage{tikz-cd}
\usepackage{bbm}
\usepackage{tikz}
\usepackage{changepage}
\usepackage{makecell}
\usepackage{hyperref}
\usepackage{algorithm,algpseudocode, caption, needspace}

\usepackage{amsthm}
\usepackage[shortlabels]{enumitem}
\newtheoremstyle{roman}
  {\topsep}   
  {\topsep}   
  {\normalfont} 
  {}          
  {\normalfont\bfseries} 
  {.}         
  { }         
  {}          

\newtheorem{assump}{Assumption}
\newtheorem{lemma}{Lemma}
\newtheorem{prop}{Proposition}
\newtheorem{corr}{Corollary}

\theoremstyle{remark}

\theoremstyle{definition}
\newtheorem{definition}{Definition}

\DeclareMathOperator*{\argmin}{argmin}

\DeclareMathOperator*{\ep}{\mathbb{E}}

\newcommand{\vvv}{\vspace{.3cm}}

\newcommand{\norm}[1]{\left\lVert#1\right\rVert}

\newcommand{\R}{\mathbb{R}}

\newcommand{\AR}{\operatorname{AR}}
\newcommand{\cset}{\mathcal C_{n,\mathrm{pAR}}}
\newcommand{\BI}{B_I}
\newcommand{\ghat}{\widehat\gamma_n}
\newcommand{\dhat}{\widehat\delta_n}
\newcommand{\pihat}{\widehat\pi_n}
\newcommand{\sighat}{\widehat\sigma_n}
\newcommand{\gdag}{\gamma^{\dagger}}
\newcommand{\Gammao}{\Gamma_0}
\newcommand{\Qn}{W_n}
\newcommand{\Dn}{D_n}
\newcommand{\epsb}{\epsilon_b}
\newcommand{\Op}{O_p}

\definecolor{styleblue}{RGB}{0,70,140}

\definecolor{recoveredgreen}{RGB}{0,105,75}

\definecolor{evergreen}{RGB}{0,102,78}

\title{Profiled Anderson--Rubin Test: Robust Inference Allowing for Direct Effects of Instruments}
\author{Jung Hyub Lee \footnote{Graduate School of Economics, University of Tokyo, 7-3-1 Hongo, Bunkyo-ku, Tokyo, Japan 113-0033, Email: \href{mailto:jhlee@e.u-tokyo.ac.jp}{jhlee@e.u-tokyo.ac.jp}.}  }
\date{\today}
\begin{document}

\maketitle
\vspace{-2em}
\begin{center}
\small
\end{center}
\begin{abstract}
Instrumental variable analyses often rely on the assumption that instruments affect the outcome only through the endogenous regressor. In many applications, researchers can defend only a plausible range for direct effects of instruments, while conventional sensitivity analyses may be unreliable when instruments are weak. This paper proposes the profiled Anderson--Rubin (pAR) test, which considers all direct effects within a prespecified range and retains a candidate effect whenever at least one admissible direct effect is consistent with the data. Under the maintained sampling assumptions, the procedure controls false rejection for each compatible candidate without requiring strong instruments. The paper provides practical methods for constructing confidence sets and distinguishes substantive bounds from bounds tied to the realized instrument design. Simulations and applications to retirement saving and returns to schooling show that the procedure resembles conventional sensitivity analysis when instruments are strong but preserves substantially more uncertainty when identification is weak. \\
\end{abstract}


\newpage
\section{Introduction}

Instrumental variable designs use excluded variables to isolate variation in an endogenous regressor. Their credibility rests on the exclusion restriction, which requires the instruments to affect the outcome only through the endogenous regressor. However, this requirement is often difficult to defend. In many applications, substantive knowledge instead supports a limited range of direct effects of instruments.

\citet{conley2012plausibly} (hereafter, CHR) formalize this view by replacing exact exclusion with a researcher-specified range or distribution for direct effects of instruments. Although their framework makes the sensitivity assumption visible, the resulting inference still depends on how each candidate effect is evaluated. Conventional Wald intervals can be unreliable when the instruments are weak. Furthermore, a grid search can also leave gaps between the values that are checked. The main focus of this paper is how to preserve the sensitivity interpretation in the spirit of CHR without requiring strong instruments or a discrete approximation.

This paper proposes the profiled Anderson--Rubin (pAR) test. For any candidate treatment effect, the observed relationships among the instruments, treatment, and outcome determine the direct effect needed to reconcile that candidate with the data. The pAR procedure considers every direct effect in the prespecified range and retains the candidate whenever the Anderson--Rubin test accepts the candidate for at least one admissible direct effect. Instead of only using a grid of selected values, therefore, inversion uses every admissible direct effect.

Under the paper's conditional Gaussian model, the Anderson--Rubin test evaluated at the compatible direct effect of instruments has its usual reference distribution. Searching over the admissible range cannot make rejection more likely. Hence, the pAR test rejects no more often than the stated level for each fixed candidate that is compatible with the range, regardless of instrument strength. 

The interpretation depends on how the range of direct effects is chosen. A range stated in substantive units yields a population compatibility analysis. A range scaled by the realized instrument design yields exact conditional inference for that design and can simplify computation. Furthermore, its connection to a stable population restriction emerges as the sample grows. Also, the method can be extended to accommodate exogenous control variables.

The paper also studies how the method behaves as the sample grows. It shows that the test rejects fixed incompatible candidates with probability approaching one, describes when the inverted confidence set is bounded, and establishes that the set approaches the population range of compatible effects under regular conditions. It also characterizes rejection near the boundary of that range. These results explain why allowing a positive range of direct effects makes the test conservative, especially when several instruments contribute to the reference cutoff.

Computation is part of the contribution. The search over admissible direct effects has a structured form that permits fast evaluation. The paper develops complementary algorithms that reuse the same preliminary work across candidate effects, while a common normalization gives a direct formula. These methods make construction of confidence set practical without relying on a grid.

The simulations examine the exact-exclusion benchmark, candidates inside and outside the admissible range, behavior near the boundary, complete inversion of confidence set, and computation. Within the maintained simulation design, the results validate the theoretical results. Exact exclusion tracks the nominal benchmark, positive ranges produce conservativeness for compatible candidates, and rejection increases as a candidate becomes incompatible with the maintained range. The alternative implementations agree numerically, and reusable preprocessing improves repeated evaluation.

The empirical illustrations show why both the sensitivity range and the component test matter. In the retirement-saving application of CHR, eligibility provides strong identifying variation, and the pAR and Wald sensitivity sets are nearly indistinguishable. In the quarter-of-birth application of \citet{angrist1991does}, the richer instrument specification provides much weaker identifying variation. The pAR confidence sets are then substantially wider than the corresponding Wald unions and expand more rapidly as the allowed direct effects increase. These comparisons suggest that changing the component test matters little with strong instruments but can significantly change the reported uncertainty with weak instruments. 


The paper proceeds as follows. Section \ref{sec:model} defines the model and establishes the conditional pAR result. Section \ref{sec:asymptotics} studies large-sample behavior. Section \ref{sec:estimation} develops the computational methods. Section \ref{sec:montecarlo} reports Monte Carlo evidence on the theoretical results. Section \ref{sec:empirical} provides empirical illustrations. The appendix provides proofs and implementation algorithms.

\subsection{Related literature}

\paragraph{Inference with weak instruments.}
Standard instrumental variable estimators and Wald tests can have severe size distortions when the instruments are weak (\citet{nelson1990distribution,bound1995problems,staiger1997instrumental}). More generally, valid confidence sets may need to be unbounded near nonidentification (\citet{gleser1987nonexistence,dufour1997some}). The Anderson--Rubin test provides an exact benchmark in the classical Gaussian model (\citet{anderson1949estimation}). The K/LM and CLR procedures can improve power by using score or conditioning information (\citet{kleibergen2002pivotal,moreira2003conditional}), and later work studies optimality, extensions, inversion, and  endogeneity parameters (\citet{kleibergen2007generalizing, moreira2009tests,mikusheva2010robust,magnusson2010inference, doko2014identification}). The pAR procedure retains the Anderson--Rubin component because its null reference law does not depend on instrument strength.

\paragraph{Violation of exclusion restriction.}
Research on exclusion violations follows several approaches. \citet{van2018beyond} use a subsample in which the instrument does not predict the endogenous regressor to inform the direct effect. Other papers obtain identification by restricting the relation between instrument strength and direct effects or by assuming that a sufficient fraction of instruments is valid (\citet{kolesar2015identification,kang2016instrumental}). Selection methods have also been adapted to shift-share designs (\citet{apfel2024relaxing}), while related work uses substantive causal pathways or independence restrictions to assess instrument validity (\citet{mellon2025rain,burauel2023evaluating}). The pAR procedure neither selects valid instruments nor assigns a prior distribution to their direct effects. It holds the admissible range fixed and changes the component test and the method used to evaluate that range.

\paragraph{Closely related literature.}

This paper builds on the support restriction approach of CHR. They relax exact exclusion by specifying admissible direct effects of instruments and taking the union of confidence intervals over that support. This paper retain this idea and replace the conventional Wald components in their implementation with Anderson–Rubin tests. \cite{masten2021salvaging} further study how exclusion restrictions can be relaxed when an instrumental variable model is falsified. The contribution of this paper relative to CHR and \citet{masten2021salvaging} is the theoretical and computational analysis of this Anderson–Rubin implementation. The paper characterizes the geometry and tail behavior of its inverted confidence set, establish convergence to the population identified interval under regular identification, and derive local rejection probabilities that explain conservativeness at regular boundaries. For computation, the paper expresses profiling as a quadratic projection onto an ellipsoid, reduce the binding profiling problem at positive radii to a scalar equation, and obtain a closed form under the aligned sample normalization. These results show how the geometry of admissible direct effects and instrument strength determine the behavior and computation of sensitivity confidence sets, while preserving the support-union logic of CHR.

\cite{wang2018sensitivity} develop sensitivity analysis based on the Anderson--Rubin test that remains valid with weak instruments and derive power calculations. Their main procedure bounds instrument invalidity in units of the structural error standard deviation and compares the unadjusted Anderson–Rubin statistic with a cutoff determined by the noncentral $F$-distribution and the admissible range. This paper instead impose a prespecified support directly on the vector of direct instrument effects and minimize the Anderson–Rubin statistic over that support while retaining a cutoff determined by the central $F$-distribution.


\section{Population membership and conditional pAR inference}
\label{sec:model}
 
 
The model has one scalar endogenous regressor and \(r\geq1\) instruments. The vector of direct instrument effects is the nuisance parameter. Exact exclusion sets this vector to zero, whereas the general analysis restricts it to a support prespecified the researcher. This formulation separates instrument relevance from the exclusion restriction.
 
\subsection{Population reduced form and identified set}
\label{sec:partially_identified}

Consider the structural and first-stage equations
$$
    \begin{aligned}    
    y_i &= x_i\beta + z_i'\gamma + \epsilon_i, \\
    x_i &= z_i'\pi + v_i,
    \end{aligned}
$$
where $x_i\in\mathbb R$ is the endogenous regressor, $z_i\in\mathbb R^r$ is the instrument vector, and $\gamma\in\mathbb R^r$ collects the direct effects of $z_i$. Stacking the observations gives
$$
    \begin{aligned}
        Y &= X\beta + Z\gamma + \epsilon, \\
        X &= Z\pi + V,
    \end{aligned}
$$
where $X\in\mathbb R^n$ and $Z\in\mathbb R^{n\times r}$. Substitution of the first stage into the outcome equation yields the reduced form
$$
    \begin{aligned}
        Y &= Z\delta + U, \\
        X &= Z\pi + V,
    \end{aligned}
$$
where $\delta=\pi\beta+\gamma$ and $U=V\beta+\epsilon$. Let $P$ denote the observable distribution of $(y_i,x_i,z_i)$, and define
\[
    W_{ZZ}(P)=E_P[z_i z_i'],
    \qquad
    \pi(P)=W_{ZZ}(P)^{-1}E_P[z_i x_i],
    \qquad
    \delta(P)=W_{ZZ}(P)^{-1}E_P[z_i y_i].
\]
Assume $W_{ZZ}(P)\succ0$, where $A\succ0$ means that $A$ is positive definite. For each candidate value $b\in\mathbb R$, define the vector of compatible direct effects
\begin{equation}       \label{eq:compatible_gamma}
    \gamma^\dagger(P,b)
    =
    \delta(P)-\pi(P)b.
\end{equation}
For any $\gamma\in\mathbb R^r$, the corresponding population moment is
\[
    E_P\!\left[z_i\{y_i-x_i b-z_i'\gamma\}\right]
    =
    W_{ZZ}(P)\{\gamma^\dagger(P,b)-\gamma\}.
\]
Since $W_{ZZ}(P)$ is nonsingular, $\gamma^\dagger(P,b)$ is the unique nuisance value that satisfies the population moment condition at $b$. Equivalently,
\[
    E_P\!\left[z_i\{y_i-x_i b-z_i'\gamma\}\right]=0_r
    \quad\Longleftrightarrow\quad
    \delta(P)=\pi(P)b+\gamma.
\]
Thus, the reduced form identifies the combination $\pi(P)b+\gamma$, but it does not identify $b$ and $\gamma$ separately.

The admissible set of direct effects (i.e., the support of $\gamma$) provides the additional restriction on $\gamma$. Let $\Gamma_0\subseteq\mathbb R^r$ be a population support prespecified by the researcher. The following definitions use this support to characterize the observationally compatible values of $b$.

\begin{definition}[Population identified set]
\label{def:pop_ident_set}
\[
    B_I(P;\Gamma_0)
    =
    \{b\in\mathbb R:\gamma^\dagger(P,b)\in\Gamma_0\}.
\]
\end{definition}

\begin{definition}[Membership null hypothesis]
\label{def:membership_null}
For a fixed candidate $b \in \mathbb R$, 
\[
    H_{0,\mathrm{pop}}^I(b;\Gamma_0):
    \gamma^\dagger(P,b)\in\Gamma_0.
\]
\end{definition}

Definitions \ref{def:pop_ident_set} and \ref{def:membership_null} depend only on the observable distribution $P$ and do not presume that $P$ determines a unique structural pair $(b,\gamma)$. The following equivalences make the membership interpretation explicit:
\[
    \begin{aligned}
    \gamma^\dagger(P,b)\in\Gamma_0
    &\iff
    b\in B_I(P;\Gamma_0) \\
    &\iff \exists\,\gamma\in\Gamma_0
    \text{ such that }
    E_P\!\left[z_i\{y_i-x_i b-z_i'\gamma\}\right]=0_r.
    \end{aligned}
\]

For the population geometry and computation below, consider the ellipsoidal support
$$    
    \Gamma_0(g,W_0)
    =
    \{\gamma\in\mathbb R^r:\gamma'W_0\gamma\le g^2\},
    \qquad
    g\ge0,
    \quad
    W_0 \succ 0,
$$
where $g$ and $W_0$ are prespecified and do not vary with $n$. The radius $g$ controls the amount of direct effect allowed, whereas $W_0$ controls the relative cost of different directions in instrument space. When $g=0$, the support is $\Gamma_0(0,W_0)=\{0_r\}$ and exact exclusion is imposed.

The scale of $W_0$ must be specified together with $g$ because for every $c>0$,
\[
    \Gamma_0(g,cW_0)=\Gamma_0(g/\sqrt c,W_0).
\]
Thus, the radius has no interpretation apart from its metric $W_0$. The membership definition and the validity result below apply more generally to any nonempty admissible set.

For the ellipsoidal support, define the squared required radius
\[
    m_{W_0}(b)
    =
    \gamma^\dagger(P,b)'W_0\gamma^\dagger(P,b)
    =
    \{\delta(P)-\pi(P)b\}'W_0\{\delta(P)-\pi(P)b\}.
\]
Therefore, the population identified set is
\[
    B_I\{P;\Gamma_0(g,W_0)\}
    =
    \{b:m_{W_0}(b)\le g^2\}.
\]
For any $a\in\mathbb R^r$, define the $W_0$ norm by $\|a\|_{W_0}=(a'W_0a)^{1/2}$, and write
\[
    a_{W_0}=\pi(P)'W_0\pi(P),
    \qquad
    b_{W_0}=\pi(P)'W_0\delta(P),
    \qquad
    d_{W_0}=\delta(P)'W_0\delta(P).
\]
If the first-stage strength is $\pi(P)=0_r$, then $m_{W_0}(b)=d_{W_0}$ for every $b$. Therefore,
\[
    B_I\{P;\Gamma_0(g,W_0)\}
    =
    \begin{cases}
        \mathbb R, & d_{W_0}\le g^2,\\[0.1cm]
        \varnothing, & d_{W_0}>g^2.
    \end{cases}
\]
This is the no-identification case. If $d_{W_0}\le g^2$, an admissible direct effect rationalizes the reduced-form outcome coefficient for every value of $b$. If $d_{W_0}>g^2$, the support restriction is falsified because changing $b$ cannot change the compatible direct effect when $\pi(P)=0_r$.

Suppose further that $a_{W_0}>0$, and define
\[
    b_{W_0}^{\circ}=\frac{b_{W_0}}{a_{W_0}},
    \qquad
    g_*
    =
    \sqrt{d_{W_0}-\frac{b_{W_0}^2}{a_{W_0}}}.
\]
Then we have
\[
    m_{W_0}(b)
    =
    a_{W_0}(b-b_{W_0}^{\circ})^2+g_*^2,
\]
and
\[
B_I\{P;\Gamma_0(g,W_0)\}
=
\begin{cases}
\varnothing, & g<g_*,\\[0.15cm]
\{b_{W_0}^{\circ}\}, & g=g_*,\\[0.15cm]
\displaystyle
\left[
 b_{W_0}^{\circ}-\frac{\sqrt{g^2-g_*^2}}{\sqrt{a_{W_0}}},\;
 b_{W_0}^{\circ}+\frac{\sqrt{g^2-g_*^2}}{\sqrt{a_{W_0}}}
\right], & g>g_*.
\end{cases}
\]
Thus, the population identified set is empty below $g_*$, a singleton at $g_*$, and a nondegenerate interval above $g_*$.

The falsification threshold $g_*$ is the smallest radius that makes the observable reduced form compatible with any value of $b$:
$$
    g_* = \min_{\widetilde b\in\mathbb R}
    \|\delta(P)-\pi(P)\widetilde b\|_{W_0}.
$$
A radius below $g_*$ falsifies the maintained support restriction. When $g>g_*$, the identified interval has length $2\sqrt{g^2-g_*^2}/\sqrt{a_{W_0}}$. The curvature $a_{W_0}$ measures how quickly the compatible direct effect changes as $b$ moves away from $b_{W_0}^{\circ}$. The interval therefore expands as $\pi(P)$ becomes small in the $W_0$ metric and reaches the no-identification case at $\pi(P)=0_r$.

The single instrument case makes this comparison transparent. Let $r=1$, $W_0>0$, and $\pi(P)\ne0$. Then $g_*=0$ and
\[
    B_I\{P;\Gamma_0(g,W_0)\}
    =
    \left[
       \frac{\delta(P)}{\pi(P)}-\frac{g}{|\pi(P)|\sqrt{W_0}},\;
       \frac{\delta(P)}{\pi(P)}+\frac{g}{|\pi(P)|\sqrt{W_0}}
    \right].
\]
Its length is $2g/\{|\pi(P)|\sqrt{W_0}\}$. A larger admissible direct effect widens the set, whereas a stronger first stage narrows it.

It is worth noting that \citet{masten2021salvaging} introduce the falsification adaptive set collects the parameter values compatible with models on the falsification frontier, which consists of the smallest nonfalsified relaxations.  The compatible radius in this paper can be interpreted as a falsification point within a prespecified family of ellipsoidal restrictions on direct effects. The main difference concerns the inferential target. Rather than aggregate identified sets along the frontier, we take the admissible set of direct effects as given and test whether each candidate treatment effect is compatible with that restriction.  


\subsection{Conditional Gaussian reduced-form model}
\label{sec:conditional_ar}

This subsection states the finite-sample argument in reduced-form terms and conditions on the realized instrument matrix $Z$. Let $Z\in\mathbb R^{n\times r}$ satisfy $\operatorname{rank}(Z)=r<n$. Conditional on $Z$, the sample obeys
$$
    \begin{aligned}
        Y&=Z\delta(P)+U, \\
        X&=Z\pi(P)+V.
    \end{aligned}
$$

\begin{assump}[Conditional Gaussian reduced form]
\label{ass:dist}
Conditional on $Z$, the rows satisfy
\[
    \begin{pmatrix}U_i\\V_i\end{pmatrix}
    \stackrel{\mathrm{ind}}{\sim}
    N\!\left[
       \begin{pmatrix}0\\0\end{pmatrix},
       \Sigma_{UV}
    \right],
    \qquad
    \Sigma_{UV} \succ 0,
    \qquad i=1,\ldots,n.
\]
\end{assump}

\begin{assump}[Fixed-support conditional asymptotics]
\label{ass:asymptotic}
Fix the reduced-form law \(P\), the number of instruments \(r\), the
covariance matrix \(\Sigma_{UV}\), the support
\(\Gamma_0(g,W_0)\), and the nominal level \(\alpha\in(0,1)\). For every \(n\), let \(Z_n\in\mathbb R^{n\times r}\) have rank \(r\), and suppose that
\(n-r\ge2\) for all sufficiently large \(n\). Along the realized
design sequence,
\[
    W_n(Z)=\frac{Z'Z}{n}
    \longrightarrow W_{ZZ}(P)\succ0.
\]
Conditional on \(Z\), the rows of \((U_n,V_n)\) are independent
Gaussian vectors with covariance matrix \(\Sigma_{UV}\succ0\).
All probability limits and distributional limits in this section
are taken under the conditional laws \(P(\,\cdot\mid Z)\) along
design sequences satisfying the displayed convergence.
\end{assump}

Assumption \ref{ass:dist} is deliberately strong because it delivers the exact finite-sample reference distribution used in this paper. Conditional homoskedastic Gaussian errors and full rank of \(Z\) yield the central \(F\) distribution at the compatible nuisance value. The unrestricted covariance matrix permits endogeneity through correlation between the reduced-form outcome and first-stage errors, and the result imposes no lower bound on \(\pi(P)\). Assumption \ref{ass:asymptotic} is used for the fixed-support limits and for relating sample-normalized and population supports. 


For a fixed candidate $b\in\mathbb R$, define the compatible reduced-form residual
\begin{equation*}
    \epsilon_b=U-Vb.
    \label{eq:eta_b}
\end{equation*}
For any candidate nuisance value $\gamma\in\mathbb R^r$, the regression residual is
\begin{equation*}
\begin{aligned}
    R_b(\gamma)
    &={Y-Xb-Z\gamma}\\
    &=\epsilon_b+Z\{\gamma^\dagger(P,b)-\gamma\}.
\end{aligned}
\label{eq:compatible_residual}
\end{equation*}
This identity holds for every tested value $b$ and does not require $b$ to equal a designated structural coefficient. Define
\[
    K=Z'Z,
    \qquad
    P_Z=ZK^{-1}Z',
    \qquad
    M_Z=I_n-P_Z,
\]
and let the sample projection coefficient be
\begin{equation}
    \begin{aligned}
    \widehat\gamma_n(b)
    &=
    K^{-1}Z'(Y-Xb) \\
    &=\gamma^\dagger(P,b)+K^{-1}Z'\epsilon_b,
    \label{eq:gammahat_compatible}
    \end{aligned}
\end{equation}
where $\epsilon_b \mid Z\sim N(0,\sigma_b^2 I_n)$ with $\sigma_b^2>0$ under Assumption \ref{ass:dist}.

For a fixed nuisance value $\gamma$, define the fixed-$\gamma$ AR statistic
$$
    AR(b;\gamma)
    =
    \frac{R_b(\gamma)'P_ZR_b(\gamma)}
    {R_b(\gamma)'M_ZR_b(\gamma)/(n-r)},
$$
where the numerator and denominator use the orthogonal decomposition
$$
    R_b(\gamma)'R_b(\gamma)
    =
    R_b(\gamma)'P_ZR_b(\gamma)
    +
    R_b(\gamma)'M_ZR_b(\gamma).
$$
The numerator is the $K$-metric discrepancy between $\widehat\gamma_n(b)$ and $\gamma$:
\begin{equation*}
\begin{aligned}
    N_n(b;\gamma)
    &={R_b(\gamma)'P_ZR_b(\gamma)}\\
    &=\{\widehat\gamma_n(b)-\gamma\}'K
      \{\widehat\gamma_n(b)-\gamma\}.
\end{aligned}
\label{eq:AR_numerator_projection}
\end{equation*}
Using \eqref{eq:gammahat_compatible},
\[
    \widehat\gamma_n(b)-\gamma\mid Z
    \sim
    N\!\left(
       \gamma^\dagger(P,b)-\gamma,
       \sigma_b^2K^{-1}
    \right),
\]
or, equivalently,
\[
    Z'R_b(\gamma)\mid Z
    \sim
    N\!\left(
       K\{\gamma^\dagger(P,b)-\gamma\},
       \sigma_b^2K
    \right).
\]
It follows that the numerator divided by $\sigma_b^2$ has a noncentral chi-squared distribution with noncentrality parameter $\lambda_b(\gamma)$:
\begin{equation}
    \frac{N_n(b;\gamma)}{\sigma_b^2}\biggm| Z
    \sim
    \chi_r^2\{\lambda_b(\gamma)\},
    \qquad
    \lambda_b(\gamma)
    =
    \frac{
      \{\gamma-\gamma^\dagger(P,b)\}'K
      \{\gamma-\gamma^\dagger(P,b)\}
    }{\sigma_b^2}.
    \label{eq:noncentrality_compatible}
\end{equation}
The denominator is
\begin{equation*}
    D_n(b)
    =
    \frac{R_b(\gamma)'M_ZR_b(\gamma)}{n-r}
    =
    \frac{\epsilon_b'M_Z\epsilon_b}{n-r}.    \label{eq:ar_denominator_compatible}
\end{equation*}
For future use, we define
$$
    \sighat (b) = D_n(b)^{1/2}.
$$
For each fixed $b\in\mathbb R$, the denominator is positive almost surely under Assumption \ref{ass:dist}. It does not depend on $\gamma$ because $M_ZZ=0$, and
\[
    \frac{(n-r)D_n(b)}{\sigma_b^2}\biggm| Z
    \sim
    \chi_{n-r}^2.
\]
The numerator depends on $P_Z\epsilon_b$, whereas the denominator depends on $M_Z\epsilon_b$. Conditional Gaussianity and $P_ZM_Z=0$ make these projections independent.

Consequently, conditional on $Z$, the $\gamma$-specific AR statistic has the noncentral $F$-distribution
\[
    AR(b;\gamma) \mid Z
    \sim
    rF_{r,n-r}\{\lambda_b(\gamma)\},
\]
where $\lambda_b(\gamma)$ is defined in \eqref{eq:noncentrality_compatible}. At the compatible nuisance value $\gamma^\dagger(P,b)$,
\begin{equation}
    \frac{AR\{b; \gamma^\dagger(P,b)\}}{r}
    \biggm| Z
    \sim
    F_{r,n-r}.
    \label{eq:central_compatible_ar}
\end{equation}
Therefore, the compatible nuisance gives the central reference statistic under the membership null $H_{0,\mathrm{pop}}^I(b;\Gamma_0)$. It is determined by the observable law and the tested value $b$, rather than by selecting an unknown structural direct effect. The exact conditional critical value is
$$
    c_{1-\alpha,n}=rF^{-1}_{r,n-r}(1-\alpha), \quad 0<\alpha<1.
$$
For fixed $r$, $rF_{r,n-r}$ converges to $\chi_r^2$ as $n$ increases. Thus, the chi-squared critical value is a large-sample approximation to the exact Gaussian reference value, and $c_{1-\alpha,n}\longrightarrow q_{r,1-\alpha}$, where $q_{r,1-\alpha}$ is the $(1-\alpha)$ quantile of $\chi_r^2$.

\subsection{Profiled Anderson--Rubin statistic}
\label{sec:uniform}
For a fixed candidate $b\in\mathbb R$ and any nonempty admissible set $\Gamma$, define the profiled Anderson--Rubin (pAR) statistic.

\begin{definition}[Profiled Anderson-Rubin statistic]
    For a fixed $b \in \mathbb R$ and a given $\Gamma$,
    \label{eq:par_definition}
    $$
        pAR(b;\Gamma) = \inf_{\gamma\in\Gamma}AR(b;\gamma).
    $$
\end{definition}
The formulation separates instrument relevance from exclusion restriction. An instrument can be informative about the endogenous regressor and still affect the outcome directly. Rather than discarding such an instrument, the pAR statistic asks whether the observed reduced form can be reconciled with a direct effect inside the maintained support. Then it reports the conclusions implied by an explicit bound on direct effects.

The pointwise inversion set based on the pAR statistic is
\begin{equation*}
    \mathcal C_{n,\mathrm{pAR}}(\Gamma)
    =
    \{b\in\mathbb R:pAR(b;\Gamma)\le c_{1-\alpha,n}\}.
    \label{eq:par_inversion}
\end{equation*}
For a fixed nuisance value $\gamma\in\mathbb R^r$, define the $\gamma$-specific AR acceptance set
\[
    \mathcal C_{\gamma,1-\alpha}^{\rm AR}
    =
    \{b:AR(b;\gamma)\le c_{1-\alpha,n}\}.
\]
When the infimum is attained, the following equivalences hold.
\begin{align*}
    b\in \mathcal C_{n,\mathrm{pAR}}(\Gamma)
    &\Longleftrightarrow
    \inf_{\gamma\in\Gamma}AR(b;\gamma)
    \le c_{1-\alpha,n}\\
    &\Longleftrightarrow
    \exists\,\gamma\in\Gamma
    \text{ such that }
    AR(b;\gamma)\le c_{1-\alpha,n}.
\end{align*}
Hence,
\begin{equation}
    \label{eq:continuous_support_union}
    \mathcal C_{n,\mathrm{pAR}}(\Gamma)
    =
    \bigcup_{\gamma\in\Gamma}
    \mathcal C_{\gamma,1-\alpha}^{\rm AR}.
\end{equation}

Equation \eqref{eq:continuous_support_union} gives the support-union interpretation of CHR. The pAR procedure replaces a conventional Wald component with an AR component and evaluates the continuous union by optimization. This change matters under weak identification because the Gaussian reference distribution does not depend on first-stage strength.


The validity argument has two steps. First, profiling cannot produce a statistic larger than the AR statistic at any admissible nuisance value. Second, under the membership null, the compatible nuisance value is admissible and gives the central conditional reference distribution in (\ref{eq:central_compatible_ar}).

\begin{lemma}[Profiling inequality]
\label{lem:par_minimum}
For every fixed $b \in \mathbb R$ and every $\widetilde\gamma\in\Gamma$,
\[
    pAR(b;\Gamma)
    \le
    AR(b;\widetilde\gamma).
\]
In particular, under $H_{0,\mathrm{pop}}^I(b;\Gamma_0)$,
\[
    pAR(b;\Gamma_0)
    \le
    AR\{b;\gamma^\dagger(P,b)\}.
\]
\end{lemma}

\begin{prop}[Pointwise finite-sample validity]
\label{prop:pointwise_validity}
\label{prop:uniform_ci}
Fix $b\in\mathbb R$. Suppose Assumption \ref{ass:dist} holds and $\Gamma_0$ is prespecified. If
\[
    H_{0,\mathrm{pop}}^I(b;\Gamma_0):
    \gamma^\dagger(P,b)\in\Gamma_0,
\]
then, 
\begin{equation*}
    \Pr_P\!\left
      \{pAR(b;\Gamma_0)>c_{1-\alpha,n}\}
      \mid Z
    \right)
    \le \alpha.
    \label{eq:pointwise_level}
\end{equation*}
Equivalently,
\begin{equation*}
    \Pr_P\{b\in \mathcal C_{n,\mathrm{pAR}}(\Gamma_0)\mid Z\}
    \ge 1-\alpha.
    \label{eq:pointwise_coverage}
\end{equation*}
The same inequalities hold unconditionally.
\end{prop}

Proposition \ref{prop:pointwise_validity} gives a conservative pointwise test. The distribution of pAR statistic can depend on the size and shape of $\Gamma_0$ and on the location of $\gamma^\dagger(P,b)$ within that set. Also, the result is pointwise in the tested value $b$. Note that it is distinct from the simultaneous coverage of the population identified set:
\[
    \Pr_P\{B_I(P;\Gamma_0)\subseteq \mathcal C_{n,\mathrm{pAR}}(\Gamma_0)\mid Z\}
    \ge1-\alpha.
\]
The simultaneous statement would require a separate argument that is uniform in $b$. 

To state its uniformity across data-generating processes, define the following membership-null class.

\begin{definition}[Membership-null class] For a fixed $b \in \mathbb R$ and a given $\Gamma_0$,
\[
    \mathcal P_0^I(b,\Gamma_0)
    =
    \left\{
       P:\ P\text{ satisfies Assumption \ref{ass:dist} and }
       \gamma^\dagger(P,b)\in\Gamma_0
    \right\}.
\]
\end{definition}
\noindent
This notation separates uniformity across data-generating processes for a fixed tested value $b$ from simultaneous coverage across values of $b$. Proposition \ref{prop:pointwise_validity} holds for every $P\in\mathcal P_0^I(b,\Gamma_0)$ and therefore implies the unconditional bound
\[
    \sup_{P\in\mathcal P_0^I(b,\Gamma_0)}
    \Pr_P\! 
      \{pAR(b;\Gamma_0)>c_{1-\alpha,n}\}
    \le\alpha,
\]
The bound follows from the event containment in Lemma \ref{lem:par_minimum}:
$$
\{pAR(b;\Gamma_0)>c_{1-\alpha,n}\}
    \subseteq
    \left\{
       AR\{b;\gamma^\dagger(P,b)\}>c_{1-\alpha,n}
    \right\}.
$$
Under Assumption \ref{ass:dist}, the event on the right has conditional probability $\alpha$. Therefore, the pAR test is conservative in general.

\subsection{Exogenous control variables}
\label{sec:exogenous_control}
Although the formal analysis in the previous subsections omits exogenous controls for simplicity, the model can include control variables. Let $c_i\in\mathbb R^k$ contain an intercept and exogenous controls and $P_C$ denote the joint observable distribution. Consider the augmented equations
\[
    \begin{split}
        y_i&=x_i\beta+z_i'\gamma+c_i'\theta+\epsilon_i, \\
        x_i&=z_i'\pi+c_i'\kappa+v_i.
    \end{split}
\]
Suppose that $Q_{CC}(P_C)=E_{P_C}[c_i c_i']\succ0$. For each scalar $q_i\in\{y_i,x_i\}$, define the residual from the population linear projection on $c_i$ by
\[
    \widetilde q_i
    =
    q_i-E_{P_C}[q_i c_i']Q_{CC}(P_C)^{-1}c_i.
\]
The same definition is applied componentwise to form $\widetilde z_i$. By linearity of the projection,
\[
    \widetilde y_i
    =
    \widetilde x_i\beta+\widetilde z_i'\gamma+\widetilde\epsilon_i,
    \qquad
    \widetilde x_i
    =
    \widetilde z_i'\pi+\widetilde v_i,
\]
where $\widetilde\epsilon_i$ and $\widetilde v_i$ are defined by the same projection. Let
\[
    \begin{aligned}
    W_{ZZ\cdot C}(P_C)
    &=E_{P_C}[\widetilde z_i\widetilde z_i'],\\
    \pi_C(P_C)
    &=[W_{ZZ\cdot C}(P_C)]^{-1}E_{P_C}[\widetilde z_i\widetilde x_i],\\
    \delta_C(P_C)
    &=[W_{ZZ\cdot C}(P_C)]^{-1}E_{P_C}[\widetilde z_i\widetilde y_i].
    \end{aligned}
\]
If $W_{ZZ\cdot C}(P_C)\succ0$, the direct effect compatible with a candidate $b$ is
\[
    \gamma_C^\dagger(P_C,b)
    =
    \delta_C(P_C)-\pi_C(P_C)b,
\]
and
\[
    E_{P_C}\!\left[
       \widetilde z_i
       \{\widetilde y_i-\widetilde x_i b-\widetilde z_i'\gamma\}
    \right]
    =
    W_{ZZ\cdot C}(P_C)\{\gamma_C^\dagger(P_C,b)-\gamma\}.
\]
Therefore, the control-adjusted identified set is
\[
    B_{I,C}(P_C;\Gamma_0)
    =
    \{b\in\mathbb R:\gamma_C^\dagger(P_C,b)\in\Gamma_0\}.
\]
Thus, controls change the residualized reduced-form coefficients but not the identified-set membership argument. The preceding ellipsoidal geometry follows after replacing $\delta(P)$ and $\pi(P)$ with $\delta_C(P_C)$ and $\pi_C(P_C)$.

The sample calculation has the same form. Stack $c_i'$ in $C\in\mathbb R^{n\times k}$, and suppose that $\operatorname{rank}(C)=k$, $\operatorname{rank}(M_C Z)=r$, and $n>k+r$, where
\[
    M_C=I_n-C(C'C)^{-1}C',
    \qquad
    K_C=Z' M_C Z,
\]
\[
    P_{Z\mid C}=M_C Z K_C^{-1}Z' M_C,
    \qquad
    M_{[C,Z]}=M_C-P_{Z\mid C}.
\]
The control-adjusted compatible coefficient and denominator are
\[
    \widehat\gamma_C(b)=K_C^{-1}Z' M_C(Y-Xb)
\]
and
\[
    D_C(b)
    =
    \frac{(Y-Xb)'M_{[C,Z]}(Y-Xb)}{n-k-r}.
\]
Accordingly, the component and profiled statistics are obtained from the subsequent formulas by replacing $K$, $P_Z$, $M_Z$, $\widehat\gamma_n(b)$, and $n-r$ with $K_C$, $P_{Z\mid C}$, $M_{[C,Z]}$, $\widehat\gamma_C(b)$, and $n-k-r$. A sample-normalized support uses $W_{n,C}=K_C/n$. Hence, changing $C$ changes the residualized instrument metric and defines a different design-dependent sensitivity restriction.

The Gaussian argument requires one additional observation. Let $H_C\in\mathbb R^{n\times(n-k)}$ have orthonormal columns spanning the orthogonal complement of $C$, so that
\[
    H_C'H_C=I_{n-k},
    \qquad
    H_C H_C'=M_C.
\]
Premultiplication by $H_C'$ removes the controls and gives the no-control model for $H_C'Y$, $H_C'X$, and $H_C'Z$. In particular, $(H_C'Z)'(H_C'Z)=K_C$, and the transformed numerator and denominator equal the quadratic forms defined by $P_{Z\mid C}$ and $M_{[C,Z]}$. If the rows of $(U,V)$ are independent homoskedastic Gaussian vectors conditional on $(C,Z)$, then the rows of $(H_C'U,H_C'V)$ remain independent Gaussian vectors with the same covariance matrix. Thus, the exact critical value is $rF^{-1}_{r,n-k-r}(1-\alpha)$. This argument uses $H_C'$, rather than treating multiplication by $M_C$ as if it preserved independent observations.

For notational simplicity, all subsequent formal results maintain $k=0$ and use the no-control notation. Under the corresponding full-rank and Gaussian conditions, the finite-sample argument applies to the transformed variables. Under the analogous convergence condition for $K_C/n$, the fixed-support asymptotic analysis also carries over. 

\section{Asymptotic results}
\label{sec:asymptotics}
This section studies the fixed population support introduced in Section \ref{sec:partially_identified} while retaining the finite-sample model and notation of Section \ref{sec:conditional_ar}. Assumptions \ref{ass:dist} and \ref{ass:asymptotic} are maintained throughout, and the admissible set is the nonrandom ellipsoid \(\Gammao(g,W_0)\). The exact pointwise bound continues to hold for every first-stage coefficient. The asymptotic analysis is used to study consistency, inversion geometry, set convergence, and local behavior at a regular boundary. All limits are conditional on the displayed design sequence. 

Define the sample reduced-form projection coefficients
\[
    \dhat=K^{-1}Z'Y,
    \qquad
    \pihat=K^{-1}Z'X.
\]
For a candidate value $b \in \mathbb R$, the sample compatible coefficient is therefore
\begin{equation*}
    \ghat(b)
    =K^{-1}Z'(Y-Xb)
    =\dhat-\pihat b.
\end{equation*}
As in Section \ref{sec:conditional_ar}, the denominator of the fixed-$\gamma$ AR statistic is $D_n(b)$, and $\sighat(b)=\Dn(b)^{1/2}$. For $A\succ0$, define the squared $A$-metric distance
\begin{equation*}
    d_A^2(x,\Gamma)
    =\inf_{\gamma\in\Gamma}(x-\gamma)'A(x-\gamma).
\end{equation*}

\begin{lemma}[Distance representation]
\label{lem:distance-denominator}
The following statements hold.
\begin{enumerate}[label=(\roman*)]
\item For every $b \in \mathbb R$ such that \(\Dn(b)>0\),
\begin{equation}
\label{eq:distance-representation}
    pAR(b;\Gammao)
    =\frac{
      n\,d_{W_n(Z)}^2\{\ghat(b),\Gammao\}
    }{
      \Dn(b)
    }.
\end{equation}
\item If \(n-r\ge2\), then, conditional on \(Z\),
\begin{equation}
\label{eq:global-denominator-positive}
    \Pr_P\!\left\{
       \Dn(b)>0\text{ for every }b\in\R
       \mid Z
    \right\}=1.
\end{equation}
Consequently, inversion of the profiled Anderson--Rubin statistic is equivalent to
\begin{equation}
\label{eq:acceptance-inequality}
    \cset(\Gammao)
    =\left\{
        b\in\R:
        n\,d_{W_n(Z)}^2\{\ghat(b),\Gammao\}
        \le c_{1-\alpha,n}\Dn(b)
      \right\}.
\end{equation}
\end{enumerate}
\end{lemma}

Because \(K=n\Qn(Z)\), the distance metric in \eqref{eq:distance-representation} is sample dependent even though the admissible set remains the fixed support \(\Gammao(g,W_0)\). Hence, the distance representation evaluates a fixed-support null with the metric induced by the numerator of AR statistic. 

\subsection{Consistency against fixed alternatives}
\begin{prop}
\label{prop:fixed-consistency}
For every fixed \(b\in\R\),
\begin{equation*}
    \frac{1}{n} pAR(b;\Gammao)
    \overset{p}{\longrightarrow}
    \frac{
      d_{W_{ZZ}(P)}^2\{\gdag(P,b),\Gammao\}
    }{
      \sigma_b^2
    }.
\end{equation*}
Consequently,
\begin{enumerate}[label=(\roman*)]
\item If \(b\notin\BI(P;\Gammao)\), then
\[
    \frac{1}{n} pAR(b;\Gammao)
    \overset{p}{\longrightarrow}
    \frac{
      d_{W_{ZZ}(P)}^2\{\gdag(P,b),\Gammao\}
    }{
      \sigma_b^2
    }>0,
\]
and
\[
    \Pr_P\!\left\{
       pAR(b;\Gammao)>c_{1-\alpha,n}
       \mid Z
    \right\}\longrightarrow1.
\]
\item If \(\gdag(P,b)\in\operatorname{int}(\Gammao)\), then
\[
    \Pr_P\!\left\{
       pAR(b;\Gammao)=0
       \mid Z
    \right\}\longrightarrow1.
\]
\end{enumerate}
\end{prop}
The limit in Proposition \ref{prop:fixed-consistency} separates the population discrepancy from the disturbance variance. If a fixed candidate $b$ lies outside the population identified set, its compatible nuisance remains a positive distance from the support and the pAR statistic grows at rate $n$. At a strict interior point, the support eventually absorbs the sampling error in $\ghat(b)$ and the profiled numerator is zero. At any fixed membership-null point, including a boundary point, the profiling inequality gives
\[
    pAR(b;\Gammao)
    \le
    \AR\{b;\gdag(P,b)\}
    =\Op(1).
\]
Thus, fixed alternatives are rejected consistently. By contrast, membership-null statistics remain stochastically bounded. Proposition \ref{prop:boundary-local-power} below gives the sharper limit at regular endpoints.

\subsection{Conservativeness of the finite-sample size bound}
\begin{prop}
\label{prop:strict-conservativeness}
Fix \(n\), a realized full-rank \(Z\), and a candidate value \(b\).  Suppose 
\[
    \gdag(P,b)\in\Gamma,
\]
where \(\Gamma\subset\R^r\) is compact and nonempty.  Then:
\begin{enumerate}[label=(\roman*)]
\item If \(\Gamma=\{\gdag(P,b)\}\), then
\[
    \Pr_P\!\left\{
       pAR(b;\Gamma)>c_{1-\alpha,n}
       \mid Z
    \right\}=\alpha.
\]
\item If \(\Gamma\) contains a point distinct from \(\gdag(P,b)\), then
\[
    \Pr_P\!\left\{
       pAR(b;\Gamma)>c_{1-\alpha,n}
       \mid Z
    \right\}<\alpha.
\]
\end{enumerate}
For the centered ellipsoid \(\Gammao(g,W_0)\), the membership-null rejection probability is exactly \(\alpha\) when \(g=0\) and is strictly below \(\alpha\) when \(g>0\).
\end{prop}

Proposition \ref{prop:strict-conservativeness} shows that nontrivial profiling is strictly conservative pointwise. Equality in the rejection-probability bound occurs only when the admissible set is the singleton $\{\gdag(P,b)\}$. A positive-radius support contains additional nuisance values, and the profiled statistic is strictly smaller than the compatible-nuisance AR statistic on an event of positive probability. The amount of conservativeness depends on the location of $\gdag(P,b)$ within the support. At a strict interior point, the rejection probability converges to zero. Also, at a regular boundary point, Proposition \ref{prop:boundary-local-power} below gives a nondegenerate limit.

\subsection{Topology and tail behavior of the inverted set}

\begin{prop}
\label{prop:topology-boundedness}
Fix a sample realization satisfying $K \succ 0$. Define
\[
    s_{X,n}^2
    =\frac{X'M_ZX}{n-r}>0,
\]
and 
$$
    \cset(\Gammao)
    =\left\{
        b\in\R:
        n\,d_{\Qn(Z)}^2\{\ghat(b),\Gammao\}
        \le c_{1-\alpha,n}\Dn(b)
      \right\}.
$$
Then:
\begin{enumerate}[label=(\roman*)]
\item \(\cset(\Gammao)\) is a closed semialgebraic subset of \(\R\).  Therefore, it is a finite union of closed bounded intervals, singleton points, and closed rays.  The empty set and the entire real line are also possible.
\item Define the tail statistic
\begin{equation*}
    L_n
    =\frac{\pihat'K\pihat}{X'M_ZX/(n-r)}.
\end{equation*}
Then
\begin{equation*}
    \lim_{|b|\to\infty} pAR(b;\Gammao)=L_n.
\end{equation*}
Consequently,
\[
    L_n>c_{1-\alpha,n}
    \quad\Longrightarrow\quad
    \cset(\Gammao)\text{ is bounded},
\]
whereas
\[
    L_n<c_{1-\alpha,n}
    \quad\Longrightarrow\quad
    \cset(\Gammao)\text{ contains both tails}.
\]
When \(L_n=c_{1-\alpha,n}\), the leading terms do not classify the tails and lower-order terms must be examined.
\end{enumerate}
\end{prop}

Proposition \ref{prop:topology-boundedness} describes the possible global shapes of $\cset(\Gammao)$ and classifies its tails. The tail statistic $L_n$ can be a useful first step because it distinguishes bounded inversion from two-sided unbounded inversion before the finite components are located. 

\begin{corr}[First-stage interpretation of the tail limit]\label{cor:tail-first-stage}
Conditional on \(Z\),
$$
    \frac{L_n}{r} \mid Z
    \sim
    F_{r,n-r}\!\left(\lambda_{\pi,n}\right),
    \qquad
    \lambda_{\pi,n} \mid Z
    =\frac{\pi(P)'K\pi(P)}{\sigma_V^2},
$$
where \(\sigma_V^2\) is the lower-right diagonal element of \(\Sigma_{UV}\).  

\begin{enumerate}[label=(\roman*)]

\item If \(\pi(P)\ne0_r\), then
\begin{equation*}
    \frac{L_n}{n} \mid Z
    \overset{p}{\longrightarrow}
    \frac{\pi(P)'W_{ZZ}\pi(P)}{\sigma_V^2}>0,
\end{equation*}
and
\[
    \Pr_P\!\left\{
       \cset(\Gammao)\text{ is bounded}
       \mid Z
    \right\}\longrightarrow1.
\]

\item If \(\pi(P)=0_r\), then
\[
    \Pr_P\!\left\{
       \cset(\Gammao)\text{ contains both tails}
       \mid Z
    \right\}=1-\alpha.
\]
\end{enumerate}
\end{corr}

\subsection{Convergence to a regular population identified interval}
Retain the population quantities from Section \ref{sec:partially_identified},
\[
    a_{W_0}=\pi(P)'W_0\pi(P),
    \qquad
    b_{W_0}=\pi(P)'W_0\delta(P),
    \qquad
    d_{W_0}=\delta(P)'W_0\delta(P),
\]
\[
    b_{W_0}^{\circ}=\frac{b_{W_0}}{a_{W_0}},
    \qquad
    g_*=\left(
       d_{W_0}-\frac{b_{W_0}^2}{a_{W_0}}
    \right)^{1/2}.
\]
For nonempty compact sets $A,B\subset\R$, let $d_H(A,B)$ denote their Hausdorff distance. Define the extended Hausdorff distance
\[
    d_H^{\mathrm{ex}}(A,B)
    =
    \begin{cases}
      d_H(A,B),&A\text{ is nonempty and compact},\\
      +\infty,&\text{otherwise},
    \end{cases}
\]
where $d_H(A,B) = \max \{ \sup_{a \in A } d(a,B), \sup_{b \in B} d(A,b) \}$.

\begin{prop}
\label{prop:hausdorff}
Suppose \(\pi(P)\ne0_r\) and \(g>g_*\).  Then
\[
    \BI(P;\Gammao)=[b_L,b_U],
\]
where
\[
    b_L=b_{W_0}^{\circ}
       -\frac{\sqrt{g^2-g_*^2}}{\sqrt{a_{W_0}}},
    \qquad
    b_U=b_{W_0}^{\circ}
       +\frac{\sqrt{g^2-g_*^2}}{\sqrt{a_{W_0}}}.
\]
Moreover,
\[
    d_H^{\mathrm{ex}}\!\left(
      \cset(\Gammao),\BI(P;\Gammao)
    \right)
    =\Op(n^{-1/2}).
\]
In particular, \(\cset(\Gammao)\) is nonempty and compact with probability approaching one.
\end{prop}

Proposition \ref{prop:hausdorff} assumes a fixed nonzero first-stage coefficient and a transverse intersection between the reduced-form line and the ellipsoid. Therefore, it is a regular identification result and is not uniform over weak first stage sequences such as $\pi_n=O(n^{-1/2})$. However, this restriction does not affect the finite-sample size bound, which remains valid for arbitrary first-stage strength. Note that Hausdorff convergence is distinct from simultaneous coverage of the population identified set. A population endpoint can be excluded with nonvanishing probability while the random set remains within an $O_p(n^{-1/2})$ Hausdorff neighborhood of the population interval. 

\subsection{Local rejection probabilities at a regular boundary}

The next lemma gives the local projection result used at a smooth ellipsoidal boundary.

\begin{lemma}[Convergence of local projection distances]\label{lem:tangent-distance}
Let \(W_0\succ0\), let \(\gamma_0'W_0\gamma_0=g^2\), and set \(a_0=W_0\gamma_0\ne0_r\).  Define
\[
    \mathcal T_n
    =\left\{
       v\in\R^r:
       2a_0'v+n^{-1/2}v'W_0v\le0
     \right\},
    \qquad
    \mathcal T
    =\{v\in\R^r:a_0'v\le0\}.
\]
If \(t_n\Rightarrow t\) and \(\Qn(Z)\to W_{ZZ}\succ0\), then
\[
    \inf_{v\in\mathcal T_n}
       (t_n-v)'\Qn(Z)(t_n-v)
    \Rightarrow
    \inf_{v\in\mathcal T}
       (t-v)'W_{ZZ}(t-v).
\]
\end{lemma}

Let \(b_0\) be either endpoint of the nondegenerate identified interval in Proposition \ref{prop:hausdorff}.  For notational convenience, write
\[
    \gamma_0=\gdag(P,b_0),
    \qquad
    a_0=W_0\gamma_0,
    \qquad
    v_0=a_0'W_{ZZ}^{-1}a_0,
\]
\[
    \sigma_0^2=\sigma_{b_0}^2,
    \qquad
    \sigma_0=(\sigma_0^2)^{1/2}.
\]

\begin{prop}[Local rejection probabilities at a regular boundary]\label{prop:boundary-local-power}
For a fixed \(h\in\R\), consider the candidate sequence
\[
    b_n=b_0+\frac{h}{\sqrt n}
\]
and define
\[
    \kappa(h)
    =\frac{
       h\,a_0'\pi(P)
     }{
       \sigma_0\sqrt{v_0}
     }.
\]
Then
\begin{equation*}
    pAR(b_n;\Gammao)
    \Rightarrow
    \{G-\kappa(h)\}_+^2,
    \qquad G\sim N(0,1),
\end{equation*}
where \(x_+=\max\{x,0\}\). Then
\begin{equation*}
    \lim_{n\to\infty}
    \Pr_P\!\left\{
       pAR(b_n;\Gammao)>c_{1-\alpha,n}
       \mid Z
    \right\}
    =1-\Phi\!\left(
       \sqrt{q_{r,1-\alpha}}+\kappa(h)
     \right),
\end{equation*}
where \(q_{r,1-\alpha}\) is the \((1-\alpha)\)-quantile of \(\chi_r^2\). The candidate sequence moves locally outside the identified interval precisely when
\[
    h\,a_0'\pi(P)<0.
\]
At the boundary itself, \(h=0\),
\begin{equation*}
    \lim_{n\to\infty}
    \Pr_P\!\left\{
       pAR(b_0;\Gammao)>c_{1-\alpha,n}
       \mid Z
    \right\}
    =1-\Phi\!\left(\sqrt{q_{r,1-\alpha}}\right)
    \le\frac{\alpha}{2}.
\end{equation*}
\end{prop}

Proposition \ref{prop:boundary-local-power} fixes the data-generating process and moves the tested candidate value $b_n$ across the population boundary at the root-\(n\) rate. Only the outward normal component of the Gaussian estimation error contributes to the first-order distance from the ellipsoid. Thus, the limit is the square of the positive part of a normal random variable rather than a full chi-squared variable. The common \(r\)-degree of freedom critical value preserves the finite-sample pointwise bound but is conservative at a regular boundary. The final bound is attained when \(r=1\) and is strict when \(r>1\).

\section{Computation of the pAR statistic}
\label{sec:estimation}
This section develops the calculations used to evaluate and invert the pAR test. Conditional on the realized instrument matrix, fixed and sample-dependent supports lead to the same deterministic projection problem, although the resulting null hypotheses have different interpretations. Section \ref{sec:kkt} derives the Karush--Kuhn--Tucker (KKT) characterization. Section \ref{subsec:spectral} gives a spectral representation for repeated inversion, and Section \ref{subsec:analytical} obtains a closed form under sample normalization. Throughout the section, write
\[
    \Gamma_n(g,W_n)
    =
    \{\gamma:\gamma'W_n\gamma\le g^2\},
    \qquad g\ge0,\quad W_n\succ0.
\]

\subsection{KKT characterization of the profiled AR statistic}
\label{sec:kkt}

Section \ref{sec:uniform} established pointwise finite-sample validity for the membership null. We now derive the profiling calculation. When $\widehat\sigma_n^2(b)=D_n(b)>0$, the fixed-$\gamma$ AR statistic is
$$
    AR(b;\gamma)
    =
    \frac{N_n(b;\gamma)}{\widehat\sigma^2_n(b)},
    \qquad
    N_n(b;\gamma)
    =
    \{\widehat\gamma_n(b)-\gamma\}'K
    \{\widehat\gamma_n(b)-\gamma\}.
$$
For a fixed candidate value $b\in\mathbb R$, the denominator does not depend on $\gamma$. The pAR calculation therefore solves
$$
    \min_{\gamma\in\mathbb R^r}
    N_n(b;\gamma)
    \quad\text{such that}\quad
    \gamma'W_n\gamma\le g^2.
$$
This convex program projects $\widehat\gamma_n(b)$ onto $\Gamma_n(g,W_n)$ in the metric induced by $K$.

Because $W_n\succ0$, the admissible set is nonempty and compact for every $g\ge0$, so the continuous objective attains its minimum. Because $K\succ0$, the objective is strictly convex and the minimizer is unique. Write $\widehat\gamma=\widehat\gamma_n(b)$. Proposition \ref{prop:kkt} characterizes this minimizer.

\begin{prop}[KKT characterization]
    \label{prop:kkt}
    For a fixed candidate value $b\in\mathbb R$, consider
    $$
        pAR\{b;\Gamma_n(g,W_n)\}
        =
        \inf_{\gamma\in\Gamma_n(g,W_n)}
        \frac{N_n(b;\gamma)}{\widehat\sigma^2_n(b)}.
    $$
    The unique minimizer $\gamma^\ast$ satisfies the following three cases.

    1. If $g=0$, then $\Gamma_n(0,W_n)=\{0\}$ and
    $$
        \gamma^\ast=0,
        \qquad
        pAR\{b;\Gamma_n(0,W_n)\}
        =
        \frac{\widehat\gamma'K\widehat\gamma}{\widehat\sigma^2_n(b)}
        =
        AR(b;0).
    $$
    If $\widehat\gamma\neq0$, no finite $\lambda$ solves
    $$
        \gamma(\lambda)'W_n\gamma(\lambda)=0,
        \qquad
        \gamma(\lambda)=(K+\lambda W_n)^{-1}K\widehat\gamma.
    $$
    Instead, $\gamma(\lambda)\to0$ only as $\lambda\to\infty$. If $\widehat\gamma=0$, any $
    \lambda \ge 0$ satisfies the KKT equation, and set $\lambda^\ast=0$ by convention.

    2. If $g>0$ and $\widehat\gamma\in\Gamma_n(g,W_n)$, then the unconstrained minimizer is feasible and
    $$
        \gamma^\ast=\widehat\gamma,
        \qquad
        \lambda^\ast=0,
        \qquad
        pAR\{b;\Gamma_n(g,W_n)\}=0.
    $$

    3. If $g>0$ and $\widehat\gamma\notin\Gamma_n(g,W_n)$, then the constraint binds. There is a unique finite multiplier $\lambda^\ast>0$ satisfying
    $$
        \phi(\lambda)
        =
        \gamma(\lambda)'W_n\gamma(\lambda)-g^2
        =
        0,
        \qquad
        \gamma(\lambda)=(K+\lambda W_n)^{-1}K\widehat\gamma.
    $$
    The minimizer is $\gamma^\ast=\gamma(\lambda^\ast)$, and
    $$
        pAR\{b;\Gamma_n(g,W_n)\}
        =
        \frac{\lambda^{\ast2}}{\widehat\sigma^2_n(b)}
        \widehat\gamma'
        W_n(K+\lambda^\ast W_n)^{-1}
        K
        (K+\lambda^\ast W_n)^{-1}
        W_n\widehat\gamma.
    $$
\end{prop}

Proposition \ref{prop:kkt} yields three computational branches. If $g=0$, the admissible set is the singleton $\{0\}$, so $\gamma^\ast=0$ and
$$
    pAR\{b;\Gamma_n(0,W_n)\}=AR(b;0).
$$
When $\widehat\gamma\ne0$, the path $\gamma(\lambda)$ approaches zero only as $\lambda\to\infty$. The exact-exclusion branch should therefore be evaluated directly.

If $g>0$ and $\widehat\gamma'W_n\widehat\gamma\le g^2$, the unconstrained minimizer is feasible. In this case, $\gamma^\ast=\widehat\gamma$ and the profiled statistic is zero. If $g>0$ and $\widehat\gamma'W_n\widehat\gamma>g^2$, the constraint binds and $\lambda^\ast$ is the unique positive root of $\phi(\lambda)$. Indeed,
$$
    \phi(0)=\widehat\gamma'W_n\widehat\gamma-g^2>0,
    \qquad
    \lim_{\lambda\to\infty}\phi(\lambda)=-g^2<0,
$$
and $\phi$ is strictly decreasing in the binding case. Thus, the KKT characterization reduces the $r$-dimensional projection to a scalar root search.

The implementation follows the same order. First compute $\widehat\gamma$ and $\widehat\sigma_n^2(b)$ and verify that the denominator is positive. Then handle $g=0$, check feasibility of $\widehat\gamma$ when $g>0$, and solve the secular equation only in the positive-radius binding case. Each evaluation of $\phi(\lambda)$ requires an $r\times r$ linear solve with $K+\lambda W_n$, so the computational cost is concentrated in that final branch.

The KKT conditions also give a dual representation. When $g>0$, Slater's condition holds because the ellipsoid has nonempty interior. When $g=0$, the feasible set is the singleton $\{0\}$ and the dual optimum may be approached only as $\lambda\to\infty$. Therefore, the exact-exclusion boundary must be handled separately.

\begin{corr}[Dual representation]
\label{cor:dual}
Define
\[
    V_n(b,g)
    =
    \min_{\gamma \in \Gamma_n(g,W_n)}
    N_n(b;\gamma)
\]
and, for \(\lambda\ge0\),
\[
    q(\lambda)
    =
    \widehat\gamma'K\widehat\gamma
    -
    \widehat\gamma'K(K+\lambda W_n)^{-1}K\widehat\gamma
    -
    \lambda g^2.
\]
Then \(q(\lambda)\le V_n(b,g)\) for every \(\lambda\ge0\).

If \(g>0\), strong duality holds and
\[
    V_n(b,g)
    =
    \max_{\lambda\ge0}q(\lambda)
    =
    q(\lambda^\ast),
\]
where \(\lambda^\ast\) is the multiplier in
Proposition~\ref{prop:kkt}. Moreover,
\[
    q'(\lambda)
    =
    \gamma(\lambda)'W_n\gamma(\lambda)-g^2.
\]

If \(g=0\), then
\[
    V_n(b,0)=\widehat\gamma'K\widehat\gamma
            =\sup_{\lambda\ge0}q(\lambda).
\]
When \(\widehat\gamma\neq0\), \(q\) is strictly increasing and the
supremum is not attained at any finite \(\lambda\). When
\(\widehat\gamma=0\), \(q(\lambda)=0\) for every \(\lambda\ge0\).
\end{corr}


The dual representation provides a lower bound during test inversion. For $g>0$ and $\lambda\ge0$, $q(\lambda)$ is no larger than the minimized numerator $N_n(b;\gamma^\ast)$. Hence
$$
    \frac{q(\lambda)}{\widehat\sigma^2_n(b)}
$$
is a lower bound on $pAR\{b;\Gamma_n(g,W_n)\}$. If, for a candidate value $b$,
$$
    q(\lambda)>c_{1-\alpha,n}\widehat\sigma^2_n(b),
$$
then
$$
    pAR\{b;\Gamma_n(g,W_n)\}>c_{1-\alpha,n}.
$$
The candidate can therefore be rejected before the full optimization is completed. The comparison must use a common scale because $q(\lambda)$ is a numerator, whereas $c_{1-\alpha,n}$ is an AR critical value.

Strong duality gives a second numerical check. For $g>0$,
$$
    N_n(b;\gamma^\ast)=q(\lambda^\ast).
$$
In the unconstrained-feasible case, $\lambda^\ast=0$ and both sides equal zero. In the binding case, $\lambda^\ast>0$ is the unique finite root of the secular equation. At $g=0$, the appropriate check is the direct primal identity
$$
    N_n(b;\gamma^\ast)=N_n(b;0)=\widehat\gamma'K\widehat\gamma,
$$
not equality at a finite multiplier.

The dual multiplier also has a local sensitivity interpretation for positive radii. Let $\rho=g^2$. Wherever the value function is differentiable, the envelope theorem gives
$$
    \frac{\partial N_n(b;\gamma^\ast)}{\partial \rho}
    =
    -\lambda^\ast,
    \qquad
    \frac{\partial pAR\{b;\Gamma_n(g,W_n)\}}{\partial g}
    =
    -\frac{2\lambda^\ast g}{\widehat\sigma^2_n(b)}.
$$
Thus, $\lambda^\ast$ is the marginal value of relaxing the squared-radius bound. In the binding case, increasing $g$ lowers the pAR statistic locally at rate $2\lambda^\ast g/\widehat\sigma_n^2(b)$. In the unconstrained-feasible case, $\lambda^\ast=0$ and a small increase in $g$ does not change the statistic. At the exact-exclusion boundary, the derivative statement is replaced by
$$
    pAR\{b;\Gamma_n(0,W_n)\}=AR(b;0).
$$

\subsection{Spectral characterization of the pAR statistic}
\label{subsec:spectral}

This subsection rewrites the KKT problem in spectral coordinates. Whitening converts the admissible ellipsoid into a Euclidean ball, and diagonalizing the numerator metric expresses the binding case through a scalar secular equation. After a one-time eigendecomposition, each evaluation of this equation requires $O(r)$ operations.

For a fixed candidate value $b\in\mathbb R$, consider the projection problem from Section \ref{sec:kkt}:
$$
    \min_{\gamma\in\mathbb R^r}
    N_n(b;\gamma)
    \quad\text{such that}\quad
    \gamma'W_n\gamma\le g^2.
$$
Let $C=W_n^{1/2}$ denote the symmetric positive-definite square root of $W_n$, and define
$$
    u=C\gamma,
    \qquad
    \widehat u=C\widehat\gamma,
    \qquad
    \widetilde K=C^{-1}KC^{-1}.
$$
The constraint becomes the Euclidean ball
$$
    \gamma'W_n\gamma=\|u\|_2^2\le g^2,
$$
and the numerator becomes
$$
    N_n(b;\gamma)
    =
    (\widehat u-u)'\widetilde K(\widehat u-u).
$$
Thus, the transformed problem projects $\widehat u$ onto the Euclidean ball $\{u:\|u\|_2\le g\}$ in the positive-definite metric $\widetilde K$.

Diagonalize the whitened numerator metric as
$$
    \widetilde K=Q\Lambda Q',
    \qquad
    \Lambda=\operatorname{diag}(\Lambda_1,\ldots,\Lambda_r),
    \qquad
    \Lambda_j>0,
$$
and write
$$
    \tilde u=Q'\widehat u.
$$
In this basis, the KKT solution shrinks each spectral component by a multiplier-dependent factor. Proposition \ref{prop:spectral_par} gives the resulting secular equation.

\begin{prop}[Spectral characterization]
    \label{prop:spectral_par}
    For a fixed candidate value $b\in\mathbb R$, consider
    $$
        pAR\{b;\Gamma_n(g,W_n)\}
        =
        \inf_{\gamma\in\Gamma_n(g,W_n)}
        \frac{N_n(b;\gamma)}{\widehat\sigma^2_n(b)}.
    $$
    The unique minimizer satisfies the following.

    1. If $g=0$, then
    \[
        u^\ast=0,
        \qquad
        \gamma^\ast=0,
    \]
    and
    \[
        N_n(b;\gamma^\ast)
        =
        \widehat u'\widetilde K\widehat u
        =
        \sum_{j=1}^r\Lambda_j\tilde u_j^2.
    \]
    If $\widehat u\neq0$, no finite multiplier solves the secular equation. The boundary solution is obtained only as $\lambda\to\infty$. 
    
    2. For finite $\lambda\ge0$, define
    \[
        u(\lambda)
        =
        Q\operatorname{diag}\left(
        \frac{\Lambda_j}{\Lambda_j+\lambda}
        \right)Q'\widehat u
    \]
    and
    \[
        \phi(\lambda)
        =
        \sum_{j=1}^r
        \left(
        \frac{\Lambda_j}{\Lambda_j+\lambda}
        \right)^2
        \tilde u_j^2
        -
        g^2.
    \]
    Then
    \[
        \phi'(\lambda)
        =
        -2\sum_{j=1}^r
        \frac{\Lambda_j^2}{(\Lambda_j+\lambda)^3}
        \tilde u_j^2
        \le0,
    \]
    with strict inequality for every finite $\lambda\ge0$ whenever $\widehat u\neq0$.

    If $g>0$ and $\|\widehat u\|_2\le g$, then
    \[
        \lambda^\ast=0,
        \qquad
        u^\ast=\widehat u,
        \qquad
        \gamma^\ast=\widehat\gamma,
        \qquad
        N_n(b;\gamma^\ast)=0.
    \]

    If $g>0$ and $\|\widehat u\|_2>g$, there is a unique finite $\lambda^\ast>0$ satisfying
    \[
        \phi(\lambda^\ast)=0.
    \]
    The minimizer is
    \[
        \gamma^\ast
        =
        C^{-1}Q
        \operatorname{diag}\left(
        \frac{\Lambda_j}{\Lambda_j+\lambda^\ast}
        \right)
        Q'\widehat u.
    \]

    3. For $g>0$, the minimized numerator is
    \[
        N_n(b;\gamma^\ast)
        =
        \sum_{j=1}^r
        \Lambda_j
        \left(
        \frac{\lambda^\ast}{\Lambda_j+\lambda^\ast}
        \right)^2
        \tilde u_j^2,
    \]
    and
    \[
        \lim_{\lambda \rightarrow \infty}
        \sum_{j=1}^r
        \Lambda_j
        \left(
        \frac{\lambda}{\Lambda_j+\lambda}
        \right)^2
        \tilde u_j^2 = \sum_{j=1}^r \Lambda_j \tilde u_j^2
    \]
    coincides with the $g=0$ case.
\end{prop}

Proposition \ref{prop:spectral_par} has the same three branches as the KKT characterization. At $g=0$, the solution is $u^\ast=\gamma^\ast=0$ and the pAR statistic equals $AR(b;0)$. If $g>0$ and $\|\widehat u\|_2\le g$, the unconstrained minimizer is feasible and the statistic is zero. If $g>0$ and $\|\widehat u\|_2>g$, the constraint binds and the unique positive root of $\phi(\lambda)=0$ determines the optimizer. The minimized statistic is
$$
    pAR\{b;\Gamma_n(g,W_n)\}
    =
    \frac{1}{\widehat\sigma_n^2(b)}
    \sum_{j=1}^r
    \Lambda_j
    \left(
    \frac{\lambda^\ast}{\Lambda_j+\lambda^\ast}
    \right)^2
    \tilde u_j^2.
$$

The spectral algorithm separates one-time calculations from candidate-specific calculations. First compute $\widetilde K=C^{-1}KC^{-1}$ and its eigendecomposition $\widetilde K=Q\Lambda Q'$. For each candidate $b$, compute
$$
    \widehat\gamma=K^{-1}Z'(Y-Xb),
    \qquad
    \widehat u=C\widehat\gamma,
    \qquad
    \tilde u=Q'\widehat u.
$$
Then apply the ordered checks $g=0$, $g>0$ with $\|\widehat u\|_2\le g$, and $g>0$ with $\|\widehat u\|_2>g$. Only the last branch requires a root search. This ordering avoids a finite-multiplier approximation at the exact-exclusion boundary and unnecessary root finding when the unconstrained minimizer is feasible.

In the positive-radius binding case, let
$$
    R=\frac{\|\widehat u\|_2}{g}>1,
    \qquad
    \Lambda_{\min}=\min_j\Lambda_j,
    \qquad
    \Lambda_{\max}=\max_j\Lambda_j.
$$
Then
$$
    \Lambda_{\min}(R-1)
    \le \lambda^\ast
    \le \Lambda_{\max}(R-1),
$$
because
$$
    \frac{\Lambda_{\min}}{\Lambda_{\min}+\lambda}\|\widehat u\|_2
    \le \|u(\lambda)\|_2
    \le
    \frac{\Lambda_{\max}}{\Lambda_{\max}+\lambda}\|\widehat u\|_2.
$$
This finite bracket permits safeguarded Newton or bisection without an open-ended search for an upper bound.


The spectral representation also describes the projection geometry. The eigenvalue $\Lambda_j$ measures the cost of a discrepancy in the $j$th orthogonal direction, whereas $\tilde u_j$ measures the sample discrepancy in that direction. The optimizer multiplies this coordinate by $\Lambda_j/(\Lambda_j+\lambda^\ast)$. A larger multiplier therefore produces stronger shrinkage toward the feasible ball.

The main cost of the spectral method is the initial eigendecomposition and its memory requirement. For dense, moderate-dimensional problems with repeated evaluations, this cost can be offset by the $O(r)$ candidate-specific calculations. For very large sparse problems, or when $K$ or $W_n$ changes frequently, the direct KKT formulation with iterative linear solvers may be preferable. Both implementations should handle $g=0$ before calling a secular-equation routine.

\subsection{Closed-form solution under sample normalization}
\label{subsec:analytical}

The KKT and spectral characterizations apply to any positive-definite metric $W_n$. A closed form is available when the admissible-set metric is proportional to the AR numerator metric. In this subsection, assume
$$
    W_n=W_n(Z)=\frac{Z'Z}{n}=\frac{K}{n}\succ0.
$$
Under this normalization, the radius $g$ has a root mean-square (RMS) interpretation. Let $||x||_A = (x'Ax)^{1/2}$. Then
$$
    \|\gamma\|_{W_n}^2
    =
    \gamma'W_n\gamma
    =
    \frac{1}{n}\|Z\gamma\|_2^2.
$$
Thus, $\Gamma_n(g,W_n)$ bounds the per-observation RMS magnitude of the direct effect $Z\gamma$.

To derive the projection, define
\[
    u=W_n^{1/2}\gamma,
    \qquad
    \widehat u=W_n^{1/2}\widehat\gamma.
\]
Since $K=nW_n$, the AR numerator becomes
\[
    (\widehat\gamma-\gamma)'K(\widehat\gamma-\gamma)
    =
    n\|u- \widehat u\|_2^2,
\]
whereas the constraint becomes
\[
    \|u\|_2\le g.
\]
Therefore, the profiled numerator is $n$ times the squared Euclidean distance from $\widehat u$ to the ball of radius $g$. The projection is radial.

\begin{prop}[Closed-form solution]
    \label{prop:par_closed}
    Suppose $W_n=K/n\succ0$, and consider
    \[
        \Gamma_n(g,W_n)
        =
        \{\gamma:\gamma'W_n\gamma\le g^2\},
        \qquad
        g\ge0.
    \]
    The pAR statistic has the closed-form representation
    \begin{equation}
        \label{prop9:closed-form}
        pAR\{b;\Gamma_n(g,W_n)\}
        =
        \frac{\big(\|\widehat\gamma\|_K-\sqrt n\,g\big)_+^2}{\widehat\sigma^2_n(b)}
        =
        \frac{n\big(\|\widehat\gamma\|_{W_n}-g\big)_+^2}{\widehat\sigma^2_n(b)},
    \end{equation}
    where $(x)_+=\max\{x,0\}$. Using the exact-exclusion AR statistic as a benchmark, the same statistic can be written as
    \begin{equation}
        \label{prop9:benchmark}
        pAR\{b;\Gamma_n(g,W_n)\}
        =
        \left(
        \sqrt{AR(b;0)}
        -
        \frac{\sqrt n\,g}{\widehat\sigma_n(b)}
        \right)_+^2.
    \end{equation}
\end{prop}

Proposition \ref{prop:par_closed} compares the sample discrepancy $\|\widehat\gamma\|_{W_n}$ with the admissible RMS radius $g$. If $\|\widehat\gamma\|_{W_n}\le g$, the unconstrained minimizer is feasible and the profiled numerator is zero:
$$
    pAR\{b;\Gamma_n(g,W_n)\}=0.
$$
A vector of admissible direct effects then exactly rationalizes the candidate value $b$.

If $\|\widehat\gamma\|_{W_n}>g$, the support cannot absorb the complete sample discrepancy. The optimizer is the radial projection of $\widehat\gamma$ onto the boundary of the RMS ball, leaving the discrepancy $\|\widehat\gamma\|_{W_n}-g$. Multiplication by $n$ converts this excess RMS discrepancy to the AR numerator scale in \eqref{prop9:closed-form},
where the second equality uses $\|\widehat\gamma\|_K=\sqrt n\,\|\widehat\gamma\|_{W_n}$.

Equation \eqref{prop9:benchmark} expresses the closed form relative to the exact-exclusion AR statistic. It uses $AR(b;0)$ only as a benchmark and does not impose exact exclusion when $g>0$. Profiling subtracts the admissible RMS radius after converting it to the square-root AR scale. When $g=0$, $\Gamma_n(0,W_n)=\{0\}$ and
\[
    pAR\{b;\Gamma_n(0,W_n)\}=AR(b;0).
\]
For $g>0$, the profiled statistic is weakly smaller because an admissible direct effect can absorb part of the sample discrepancy. The adjustment $\sqrt n\,g/\widehat\sigma_n(b)$ converts the RMS radius first to the unnormalized $K$-metric scale and then to the square-root statistic scale.

The closed form simplifies computation and interpretation. For each candidate $b$, evaluation requires only $\widehat\gamma$, $\widehat\sigma_n^2(b)$, and $\|\widehat\gamma\|_{W_n}$; no root search is needed. The candidate is fully accommodated when the sample discrepancy is no larger than $g$, and the statistic penalizes only the excess.

The result also clarifies the KKT characterization. When $W_n=K/n$, the support is aligned with the AR metric and the projection is radial. The KKT multiplier then has a closed form in the positive-radius binding case. At $g=0$, the solution remains a direct projection onto the singleton rather than a finite-multiplier solution.

The spectral representation in Section \ref{subsec:spectral} reduces to the same formula. When $W_n=K/n$, $C=W_n^{1/2}$ and
$$
    \widetilde K=C^{-1}KC^{-1}=nI.
$$
All eigenvalues therefore equal $n$. For $g>0$ and $\|\widehat u\|_2>g$, the secular equation becomes
\[
    \left(
    \frac{n}{n+\lambda}
    \right)^2
    \|\widehat u\|_2^2
    =
    g^2,
\]
which gives
\[
    \lambda^\ast
    =
    n\left(
    \frac{\|\widehat u\|_2}{g}-1
    \right).
\]
The shrinkage path is radial:
\[
    u(\lambda^\ast)
    =
    \frac{g}{\|\widehat u\|_2}\widehat u.
\]
Since $\|\widehat u\|_2=\|\widehat\gamma\|_{W_n}$, the minimized numerator is
$$
    N_n(b;\gamma^\ast)
    =
    n\left(\|\widehat\gamma\|_{W_n}-g\right)_+^2
    =
    \left(\|\widehat\gamma\|_K-\sqrt n\,g\right)_+^2.
$$
At $g=0$, the radial-thresholding expression is evaluated directly; the multiplier formula is not used because it divides by $g$.

\subsection{Population versus sample-dependent admissible sets}
\label{sec:supports}

This subsection distinguishes the sample-dependent support in Section \ref{subsec:analytical} from a fixed population support. The sample-dependent support aligns its metric with the AR numerator and gives exact design-conditional inference for the null in its associated identified set. Although a population RMS restricted admissible support is a different object, the two supports are connected through the asymptotic argument under Assumptions \ref{ass:dist} and \ref{ass:asymptotic} below.

First consider the sample-dependent support $\Gamma_n(g,W_n(Z))$ and its associated identified set:
\begin{equation*}
    B_{I,n}\{P;g,W_n(Z)\}
    =
    \left\{
        b\in\mathbb R:
        \gamma^\dagger(P,b)'W_n(Z)\gamma^\dagger(P,b)\le g^2
    \right\}.
\end{equation*}
Both objects depend on the realized instrument matrix $Z$. Under Assumption \ref{ass:dist}, condition on a realization of $Z$ for which
$$
    \gamma^\dagger(P,b)\in\Gamma_n(g,W_n(Z)).
$$
Then, the same profiling and centrality argument gives
\begin{equation}
    \Pr_P\!\left\{
        pAR\{b;\Gamma_n(g,W_n(Z))\}
        >
        c_{1-\alpha,n}
        \mid Z
    \right\}
    \le\alpha.
    \label{eq:design_conditional_level}
\end{equation}
Thus, $W_n(Z)$ is compatible with exact conditional inference, but it changes the null being tested. The procedure concerns a sample-dependent object rather than a fixed population object. More generally, the same conditional argument applies to $\Gamma_Z=\Gamma(g_Z,W_Z)$ when $g_Z$ and $W_Z$ are fixed or measurable with respect to $Z$.


For a fixed observable data-generating process $P$, define the population admissible set
\begin{equation*}
    \Gamma_{\mathrm{pop}}(P;g)
    =
    \left\{
        \gamma\in\mathbb R^r:
        \gamma'W_{ZZ}(P)\gamma\le g^2
    \right\}.
\end{equation*}
Since
\[
    \gamma'W_{ZZ}(P)\gamma
    =
    E_P[(z_i'\gamma)^2],
\]
the restriction bounds the population RMS of the direct effect $z_i'\gamma$. Its associated identified set is
\[
    B_{I,\mathrm{pop}}(P;g)
    =
    \left\{
        b\in\mathbb R:
        \gamma^\dagger(P,b)'W_{ZZ}(P)
        \gamma^\dagger(P,b)\le g^2
    \right\}.
\]
Although $\Gamma_{\mathrm{pop}}(P;g)$ is nonrandom for a fixed $P$, it varies across data-generating laws. The statements below are therefore pointwise in a fixed law $P$ and a fixed candidate $b$.

Write
\[
    \gamma_b=\gamma^\dagger(P,b),
    \qquad
    W=W_{ZZ}(P),
\]
and define the event
\[
    A_n(b)
    =
    \left\{
        \gamma_b'W_n(Z)\gamma_b\le g^2
    \right\}.
\]
Call $b$ strictly feasible for the population RMS restriction when
\[
    \gamma_b'W\gamma_b<g^2.
\]
This condition is stronger than requiring $b$ to be an interior point of $B_{I,\mathrm{pop}}(P;g)$, because interiority need not imply strict feasibility when $\pi(P)=0_r$.

If
\[
    \|W_n(Z)-W\|_{\mathrm{op}}
    \overset{p}{\longrightarrow}0,
\]
then
\[
    \left|
       \gamma_b'W_n(Z)\gamma_b-\gamma_b'W\gamma_b
    \right|
    \le
    \|\gamma_b\|^2
    \|W_n(Z)-W\|_{\mathrm{op}}
    =
    o_p(1).
\]
Strict feasibility therefore implies
\[
    \Pr_P\{A_n(b)\}\longrightarrow1.
\]
Let
\[
    R_n(b)
    =
    \left\{
        pAR\{b;\Gamma_n(g,W_n(Z))\}
        >
        c_{1-\alpha,n}
    \right\}.
\]
Since \(A_n(b)\) is measurable with respect to \(Z\), the conditional
bound in \eqref{eq:design_conditional_level} gives
\[
\begin{aligned}
    \Pr_P\{R_n(b)\}
    &\le
    \alpha\Pr_P\{A_n(b)\}
    +
    \Pr_P\{A_n(b)^c\} \\
    &=
    \alpha
    +
    (1-\alpha)\Pr_P\{A_n(b)^c\}.
\end{aligned}
\]
Consequently,
\[
    \limsup_{n\to\infty}\Pr_P\{R_n(b)\}\le\alpha.
\]
Under Assumptions \ref{ass:dist} and \ref{ass:asymptotic},
\[
    \widehat\gamma_n(b)'W_n(Z)\widehat\gamma_n(b)
    \overset{p}{\longrightarrow}
    \gamma_b'W\gamma_b.
\]
Hence, at a strictly feasible point,
\[
    \Pr_P\left\{
        pAR\{b;\Gamma_n(g,W_n(Z))\}=0
    \right\}
    \longrightarrow1.
\]
The design-dependent support eventually contains the unrestricted sample coefficient, so its rejection probability converges to zero rather than merely remaining below $\alpha$.

Suppose instead that the population support restriction binds:
\[
    \gamma_b'W\gamma_b=g^2.
\]
For a nontrivial boundary point with \(g>0\),
\[
    A_n(b)
    =
    \left\{
        \gamma_b'\{W_n(Z)-W\}\gamma_b\le0
    \right\}.
\]
Convergence of $W_n(Z)$ controls the magnitude of this expression but not its sign. Consistency of $W_n(Z)$ and the containment argument therefore do not establish exact or asymptotic size control at a positive-radius population boundary. A separate joint boundary limit or an outer-support construction is required. When $g=0$, population membership implies $\gamma_b=0_r$, so $A_n(b)$ occurs for every sample.

One outer construction replaces the sample support by
\[
    \Gamma_n^+(g;Z)
    =
    \left\{
        \gamma:
        \gamma'W_n(Z)\gamma\le g^2+\kappa_n
    \right\},
\]
where \(\kappa_n>0\) is deterministic,
\[
    \kappa_n\to 0,
    \qquad
    \|W_n(Z)-W\|_{\mathrm{op}}=o_p(\kappa_n).
\]
Because \(\Gamma_{\mathrm{pop}}(P;g)\) is compact,
\[
    \Pr_P\left\{
        \Gamma_{\mathrm{pop}}(P;g)
        \subseteq
        \Gamma_n^+(g;Z)
    \right\}
    \to 1.
\]
The enlarged support retains finite-sample conditional size control for its own sample-dependent membership null. For every fixed $b\in B_{I,\mathrm{pop}}(P;g)$,
\[
    \limsup_{n\to\infty}
    \Pr_P\left\{
        pAR\{b;\Gamma_n^+(g;Z)\}
        >
        c_{1-\alpha,n}
    \right\}
    \le\alpha.
\]
Thus, it also establish a pointwise asymptotic coverage. 
\section{Monte Carlo evidence}
\label{sec:montecarlo}

This section validates the theoretical results provided in Section \ref{sec:asymptotics} using Monte Carlo simulations. First, it evaluates pointwise rejection probabilities at fixed members of the population identified set and compares pAR with oracle AR, naive exact-exclusion AR, and a Wald support union. Second, it studies fixed and local membership alternatives. Third, it inverts pAR over the complete real line and compares the resulting tail behavior with its exact finite-sample benchmark. Finally, it examines convergence of the complete random set and the numerical implementations developed in Section \ref{sec:estimation}.

\subsection{Simulation Design}

Consider the model
\begin{equation*}
    X=Z\pi_n+V,
    \qquad
    Y=X\beta_0+Z\gamma_0+\epsilon,
\end{equation*}
where $\beta_0=1$ and
\[
    \begin{pmatrix}\epsilon_i\\V_i\end{pmatrix}
    \overset{\mathrm{iid}}{\sim}
    N\!\left[
      \begin{pmatrix}0\\0\end{pmatrix},
      \begin{pmatrix}1&\rho\\ \rho&1\end{pmatrix}
    \right],
    \qquad \rho=0.5.
\]
For each pair $(n,r)$, the instrument matrix is constructed once and held fixed across replications. The construction satisfies
\[
    \frac{Z'Z}{n}=\Omega_r,
    \qquad
    (\Omega_r)_{jk}=0.5^{|j-k|}.
\]
Thus, the simulations evaluate the conditional finite-sample result under a fixed full-rank design rather than averaging over different instrument matrices.

The admissible set is the nonrandom ellipsoid
\[
    \Gammao(g,\Omega_r)
    =
    \{\gamma\in\mathbb R^r:\gamma'\Omega_r\gamma\le g^2\}.
\]
Let $d_\pi'\Omega_r d_\pi=1$. The main experiments use the aligned direction $d_\gamma=d_\pi$ and set
\[
    \gamma_0=\tau g d_\gamma.
\]
For $g=0.20$, the values $\tau=0$ and $\tau=0.5$ place $\beta_0$ in the interior of the identified set, $\tau=1$ places $\beta_0$ on its boundary, and $\tau>1$ gives a membership alternative. The exact-exclusion design sets $g=0$ and $\gamma_0=0$.

For the pointwise and weak-identification exercises, first-stage strength is indexed by
\[
    F^\ast
    =
    \frac{\pi_n'Z'Z\pi_n}{r\sigma_V^2},
    \qquad
    \pi_n=\sqrt{\frac{rF^\ast}{n}}\,d_\pi,
\]
where $\sigma_V^2=1$. We use $n\in\{100,500\}$, $r\in\{1,4\}$, and $F^\ast\in\{0,1,10\}$ in the pointwise exercise, with $n=2000$ added at the boundary. The complete inversion exercise also includes $F^\ast\in\{5,20\}$. The local-boundary and set-convergence exercises instead hold the first stage fixed at $\pi=0.4d_\pi$, as required by Propositions \ref{prop:hausdorff} and \ref{prop:boundary-local-power} .

The comparisons include the fixed-support pAR test, the infeasible oracle AR test evaluated at $\gamma_0$, and the naive AR test that imposes $\gamma=0$. For $r=1$, we also report the CHR-style union of conventional Wald intervals over $\gamma\in[-g,g]$. The AR-based procedures use $c_{1-\alpha,n}=rF^{-1}_{r,n-r}(1-\alpha)$, while the Wald union uses the conventional normal critical value. Standard rejection-probability cells use 20,000 replications, exact-size and boundary-calibration cells use 100,000 replications, and the inversion and Hausdorff exercises use 5,000 replications. Common random numbers are used within each design cell.

\subsection{Pointwise validity and conservativeness}

Table \ref{tab:mc-pointwise} reports rejection probabilities at $b=\beta_0$. At this candidate value,
\[
    Y-X\beta_0=Z\gamma_0+\epsilon,
\]
so the pAR, oracle AR, and naive AR rejection probabilities do not depend on $\pi_n$. Thus, their entries are reported once for each $(r,n,g,\tau)$ cell. The benchmark is $0.05$ under exact exclusion, zero at a strict interior point as $n$ increases, and
\[
    1-\Phi\!\left(\sqrt{q_{r,0.95}}\right)
\]
at a regular boundary.

\begin{table}[t]
\centering
\caption{Pointwise rejection probabilities at $b=\beta_0$}
\label{tab:mc-pointwise}
\small
\setlength{\tabcolsep}{4pt}
\begin{tabular}{rrlccccc}
\toprule
$n$ & $\tau$ & Support location & Benchmark & pAR & Oracle AR & Naive AR & $\Pr(\mathrm{pAR}=0)$ \\
\midrule
\multicolumn{8}{l}{\textit{1. $r=1$}}\\
\addlinespace[2pt]
100 & -- & Singleton & 0.0500 & 0.0493 & 0.0493 & 0.0493 & 0.0000 \\
500 & -- & Singleton & 0.0500 & 0.0509 & 0.0509 & 0.0509 & 0.0000 \\
100 & 0.0 & Center & 0.0000 & 0.0001 & 0.0510 & 0.0510 & 0.9545 \\
500 & 0.0 & Center & 0.0000 & 0.0000 & 0.0486 & 0.0486 & 1.0000 \\
100 & 0.5 & Interior & 0.0000 & 0.0014 & 0.0495 & 0.1682 & 0.8393 \\
500 & 0.5 & Interior & 0.0000 & 0.0000 & 0.0510 & 0.6077 & 0.9875 \\
100 & 1.0 & Boundary & 0.0250 & 0.0249 & 0.0495 & 0.5075 & 0.5007 \\
500 & 1.0 & Boundary & 0.0250 & 0.0248 & 0.0508 & 0.9937 & 0.4947 \\
2000 & 1.0 & Boundary & 0.0250 & 0.0255 & 0.0504 & 1.0000 & 0.5032 \\
\midrule
\multicolumn{8}{l}{\textit{2. $r=4$}}\\
\addlinespace[2pt]
100 & -- & Singleton & 0.0500 & 0.0492 & 0.0492 & 0.0492 & 0.0000 \\
500 & -- & Singleton & 0.0500 & 0.0497 & 0.0497 & 0.0497 & 0.0000 \\
100 & 0.0 & Center & 0.0000 & 0.0000 & 0.0500 & 0.0500 & 0.5999 \\
500 & 0.0 & Center & 0.0000 & 0.0000 & 0.0532 & 0.0532 & 0.9994 \\
100 & 0.5 & Interior & 0.0000 & 0.0001 & 0.0494 & 0.0997 & 0.4679 \\
500 & 0.5 & Interior & 0.0000 & 0.0001 & 0.0510 & 0.3905 & 0.9628 \\
100 & 1.0 & Boundary & 0.0010 & 0.0046 & 0.0500 & 0.3058 & 0.2182 \\
500 & 1.0 & Boundary & 0.0010 & 0.0022 & 0.0503 & 0.9631 & 0.3658 \\
2000 & 1.0 & Boundary & 0.0010 & 0.0018 & 0.0512 & 1.0000 & 0.4319 \\
\bottomrule
\end{tabular}
\vspace{2pt}

\begin{minipage}{0.96\textwidth}
\footnotesize
\textit{Notes:} The singleton rows use $g=0$. All other rows use $g=0.20$. The benchmark is exact for the singleton and asymptotic for the interior and boundary rows. Singleton and boundary cells use 100,000 replications. Strict-interior cells use 20,000 replications. The largest Monte Carlo standard error for a pAR rejection probability is 0.0016, and the largest standard error among all reported proportions is 0.0036. Entries are rounded to four decimal places.
\end{minipage}
\end{table}

The singleton results track the exact 5 percent benchmark. Across $r\in\{1,4\}$ and $n\in\{100,500\}$, the pAR rejection probability ranges from 0.0492 to 0.0509 and coincides with the oracle and naive AR procedures. Positive-radius profiling changes this comparison. At the center and at $\tau=0.5$, the pAR rejection probability is at most 0.0014 and is essentially zero by $n=500$. The probability that the profiled statistic equals zero also rises toward one. For example, at $\tau=0.5$ it rises from 0.8393 to 0.9875 when $r=1$ and from 0.4679 to 0.9628 when $r=4$. These results are consistent with Propositions \ref{prop:fixed-consistency} and \ref{prop:strict-conservativeness}.

The boundary results isolate the source of conservativeness. When $r=1$, the rejection probabilities are 0.0249, 0.0248, and 0.0255 for $n=100$, $500$, and $2000$, respectively. These values closely match the limiting probability of 0.025. When $r=4$, the rejection probability falls from 0.0046 to 0.0022 and then to 0.0018 as the sample size increases. The sequence moves toward the limiting probability of 0.0010, although the approximation remains incomplete at $n=2000$. This difference across $r$ is consistent with the boundary theory because only the outward normal component enters the first-order distance while the test retains an $r$-degree-of-freedom critical value.

The oracle AR rejection probability remains close to 0.05 in every cell. In contrast, the naive AR procedure rejects admissible direct effects because it imposes exact exclusion. At $\tau=0.5$ and $n=500$, its rejection probability is 0.6077 for $r=1$ and 0.3905 for $r=4$. At the boundary, the corresponding probabilities are 0.9937 and 0.9631. The Wald support union behaves differently from both procedures. Under exact exclusion with $r=1$ and $n=100$, its rejection probability rises from 0.0156 to 0.0439 as $F^\ast$ increases from zero to 10. The analogous range is 0.0145 to 0.0445 when $n=500$. Hence, the Wald union is highly conservative under a weak first stage and its behavior depends on first-stage strength, whereas the pAR rejection probability at $\beta_0$ is invariant to $F^\ast$.

\subsection{Fixed and local membership alternatives}

Table \ref{tab:mc-fixed} reports the pAR rejection probabilities. The fixed-alternative exercise keeps $b=\beta_0$ and sets $\tau\in\{1.10,1.25,1.50\}$. Since $\gamma_0\notin\Gammao$ in each case, these cells test power against the membership null rather than power against the structural equality alone. 

\begin{table}[t]
\centering
\caption{Rejection probabilities under fixed membership alternatives}
\label{tab:mc-fixed}
\small
\setlength{\tabcolsep}{9pt}
\begin{tabular}{ccccc}
\toprule
$r$ & $n$ & $\tau=1.10$ & $\tau=1.25$ & $\tau=1.50$ \\
\midrule
1 & 100 & 0.0401 & 0.0691 & 0.1669 \\
1 & 500 & 0.0628 & 0.1970 & 0.6089 \\
1 & 2000 & 0.1440 & 0.6081 & 0.9938 \\
\addlinespace[2pt]
4 & 100 & 0.0085 & 0.0149 & 0.0441 \\
4 & 500 & 0.0083 & 0.0411 & 0.2635 \\
4 & 2000 & 0.0208 & 0.2377 & 0.9357 \\
\bottomrule
\end{tabular}
\vspace{2pt}

\begin{minipage}{0.88\textwidth}
\footnotesize
\textit{Notes:} Entries are pAR rejection probabilities based on 20,000 replications. The design uses $g=0.20$, $\gamma_0=\tau g d_\pi$, and $b=\beta_0$. The largest Monte Carlo standard error is 0.0035.
\end{minipage}
\end{table}

Rejection increases with the distance from the support and, for the more separated alternatives, with the sample size. At $\tau=1.50$, the rejection probability rises from 0.1669 to 0.9938 when $r=1$ and from 0.0441 to 0.9357 when $r=4$ as $n$ increases from 100 to 2000. Alternatives close to the boundary remain difficult. At $\tau=1.10$ and $n=2000$, the rejection probabilities are 0.1440 and 0.0208 for $r=1$ and $r=4$, respectively. In these aligned designs, rejection is lower when $r=4$, which is consistent with the more conservative boundary behavior in Table \ref{tab:mc-pointwise}.

The oracle AR statistic is not a power benchmark for this exercise. It evaluates the true direct effect and therefore continues to test the structural equality $b=\beta_0$, which is correct in the data-generating process. Its rejection probabilities remain between 0.0485 and 0.0539. The naive AR rejection probabilities are much larger, but they arise from imposing the false restriction $\gamma_0=0$. The relevant evidence for Proposition \ref{prop:fixed-consistency} is the increase in pAR rejection as a fixed compatible direct effect moves farther outside the support.

The local exercise starts from the regular boundary point $\gamma_0=gd_\pi$ and the fixed first stage $\pi=pd_\pi$, where $p=0.4$. For
\[
    d\in\{0,0.5,1,1.5,2,3\},
    \qquad
    b_n(d)=\beta_0-\frac{d}{p\sqrt n},
\]
Proposition \ref{prop:boundary-local-power} gives the limiting rejection probability
\[
    1-\Phi\!\left(\sqrt{q_{r,0.95}}-d\right).
\]
Figure \ref{fig:mc-local} compares this limit with the simulated rejection probabilities.

\begin{figure}[t]
\centering
\caption{Local rejection probabilities at a regular boundary}
\includegraphics[width=0.84\textwidth]{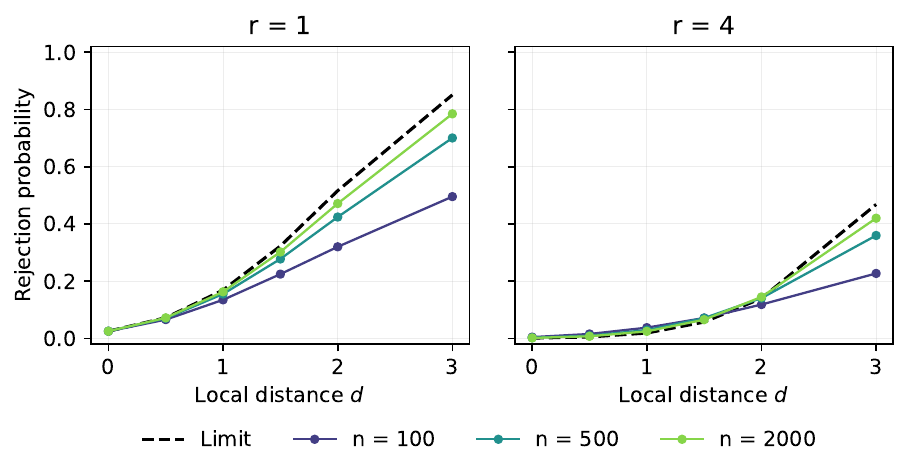}
\label{fig:mc-local}
\begin{minipage}{0.90\textwidth}
\footnotesize
\textit{Notes:} The candidate is $b_n(d)=\beta_0-d/(0.4\sqrt n)$. The dashed curve is $1-\Phi(\sqrt{q_{r,0.95}}-d)$. Each simulated point uses 100,000 replications. The largest Monte Carlo standard error is 0.0016.
\end{minipage}
\end{figure}

At $d=0$, the $r=1$ rejection probabilities range from 0.0248 to 0.0253 and reproduce the boundary benchmark. For $r=4$, the rejection probability falls from 0.0044 at $n=100$ to 0.0016 at $n=2000$, compared with a limit of 0.0010. Moving outward increases rejection in every sample. The discrepancy from the limiting curve generally narrows with $n$. At $d=3$, for example, the $r=1$ rejection probability rises from 0.4954 to 0.7847 as the limit remains 0.8508. The corresponding $r=4$ values rise from 0.2270 to 0.4200, compared with a limit of 0.4680. The figure therefore supports the direction and shape of the local-power result, while also showing that the approximation can remain conservative at distant local alternatives in moderate samples.

\subsection{Inversion under weak identification}

The inversion exercise sets $\gamma_0=0$, $g=0.20$, and $\pi_n=\sqrt{rF^\ast/n}d_\pi$. The pAR test is inverted over the complete real line. The algorithm retains all real roots of the quartic boundary equation, checks the original unsquared equality, and uses the statistic $L_n$ in Proposition \ref{prop:topology-boundedness} to classify the tails.

Corollary \ref{cor:tail-first-stage} gives an exact benchmark. If $F_{\mathrm{ncf}}(\cdot;\nu_1,\nu_2,\lambda)$ denotes the noncentral $F$-distribution function, then
\begin{equation}
\label{eq:mc-tail-benchmark}
    \Pr\!\left\{
       \cset(\Gammao)\text{ contains both tails}
       \mid Z
    \right\}
    =
    F_{\mathrm{ncf}}\!\left(
      \frac{c_{1-\alpha,n}}{r};
      r,n-r,rF^\ast
    \right).
\end{equation}
Figure \ref{fig:mc-tail} plots the empirical probability and this exact value, and Table \ref{tab:mc-inversion} reports the associated set topology. The reported acceptance probability concerns the fixed candidate $\beta_0$.

\begin{figure}[t]
\centering
\caption{Probability that the inverted pAR set contains both tails}
\includegraphics[width=0.84\textwidth]{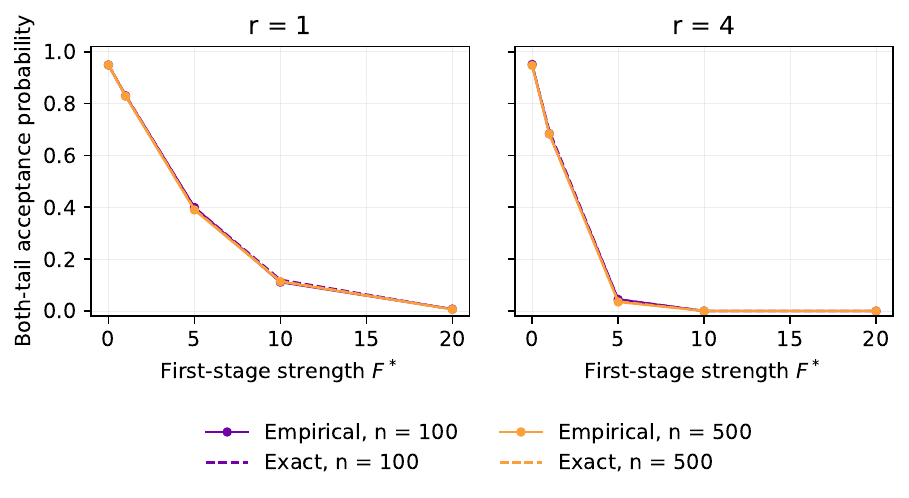}
\label{fig:mc-tail}
\begin{minipage}{0.90\textwidth}
\footnotesize
\textit{Notes:} The empirical probabilities use 5,000 replications. The exact curves use \eqref{eq:mc-tail-benchmark}. Monte Carlo standard errors are at most 0.0071.
\end{minipage}
\end{figure}

\begin{table}[t]
\centering
\caption{Topology of the inverted pAR set}
\label{tab:mc-inversion}
\small
\setlength{\tabcolsep}{6pt}
\begin{tabular}{rccccc}
\toprule
$F^\ast$ & Exact tail & Both tails & Nonempty compact & Disconnected & Whole line \\
\midrule
\multicolumn{6}{l}{\textit{1. $r=1$ and $n=100$}}\\
\addlinespace[2pt]
0 & 0.9500 & 0.9490 & 0.0510 & 0.0002 & 0.9488 \\
1 & 0.8323 & 0.8300 & 0.1700 & 0.0014 & 0.8286 \\
5 & 0.3996 & 0.3994 & 0.6006 & 0.0086 & 0.3908 \\
10 & 0.1207 & 0.1114 & 0.8886 & 0.0072 & 0.1042 \\
20 & 0.0068 & 0.0074 & 0.9926 & 0.0018 & 0.0056 \\
\midrule
\multicolumn{6}{l}{\textit{2. $r=1$ and $n=500$}}\\
\addlinespace[2pt]
0 & 0.9500 & 0.9488 & 0.0512 & 0.0000 & 0.9488 \\
1 & 0.8304 & 0.8280 & 0.1720 & 0.0000 & 0.8280 \\
5 & 0.3929 & 0.3898 & 0.6102 & 0.0000 & 0.3898 \\
10 & 0.1158 & 0.1136 & 0.8864 & 0.0000 & 0.1136 \\
20 & 0.0061 & 0.0062 & 0.9938 & 0.0000 & 0.0062 \\
\midrule
\multicolumn{6}{l}{\textit{3. $r=4$ and $n=100$}}\\
\addlinespace[2pt]
0 & 0.9500 & 0.9508 & 0.0492 & 0.0010 & 0.9498 \\
1 & 0.6945 & 0.6838 & 0.3162 & 0.0034 & 0.6804 \\
5 & 0.0450 & 0.0446 & 0.9554 & 0.0106 & 0.0340 \\
10 & 0.0003 & 0.0000 & 1.0000 & 0.0000 & 0.0000 \\
20 & 0.0000 & 0.0000 & 1.0000 & 0.0000 & 0.0000 \\
\midrule
\multicolumn{6}{l}{\textit{4. $r=4$ and $n=500$}}\\
\addlinespace[2pt]
0 & 0.9500 & 0.9468 & 0.0532 & 0.0000 & 0.9468 \\
1 & 0.6828 & 0.6828 & 0.3172 & 0.0000 & 0.6828 \\
5 & 0.0376 & 0.0350 & 0.9650 & 0.0000 & 0.0350 \\
10 & 0.0002 & 0.0002 & 0.9998 & 0.0000 & 0.0002 \\
20 & 0.0000 & 0.0000 & 1.0000 & 0.0000 & 0.0000 \\
\bottomrule
\end{tabular}
\vspace{2pt}

\begin{minipage}{0.96\textwidth}
\footnotesize
\textit{Notes:} Each cell uses 5,000 replications. ``Exact tail'' is the probability in \eqref{eq:mc-tail-benchmark}. Monte Carlo standard errors are at most 0.0071. The pointwise inclusion probability for $\beta_0$ is 0.9998 in the $r=1$, $n=100$ cells and 1.0000 in the remaining cells. No empty set occurs in these simulations. ``Whole line'' is a subset of ``Both tails.''
\end{minipage}
\end{table}

Figure \ref{fig:mc-tail} demonstrates that the empirical tail probabilities closely track the exact benchmark. The largest absolute discrepancy across the 20 cells is 0.0107. Under exact nonidentification, $F^\ast=0$, the empirical probability lies between 0.9468 and 0.9508, compared with the exact value of 0.95. The probability then declines with first-stage strength. When $r=1$, it is about 0.39 at $F^\ast=5$ and about 0.006 at $F^\ast=20$. When $r=4$, it is about 0.04 at $F^\ast=5$ and is essentially zero by $F^\ast=10$. 

The topology results reported in Table \ref{tab:mc-inversion} add information that cannot be obtained from pointwise rejection alone. At $F^\ast=0$, almost every set that contains both tails is the entire real line. As $F^\ast$ increases, the probability of a nonempty compact set rises toward one. Disconnected sets occur in finite samples, with a largest observed probability of 0.0106, and no empty set occurs. These findings support the tail classification in Proposition \ref{prop:topology-boundedness}.  

\subsection{Set convergence and computation}

The set-convergence exercise fixes $\pi=0.4d_\pi$, $\gamma_0=0$, and $g=0.20$. In the aligned design,
\[
    \BI(P;\Gammao)=[0.5,1.5].
\]
Table \ref{tab:mc-hausdorff} reports the probability that the pAR set is nonempty and compact and, conditional on this event, quantiles of
\[
    \sqrt n\,d_H\!\left\{\cset(\Gammao),\BI(P;\Gammao)\right\}.
\]

\begin{table}[t]
\centering
\caption{Convergence of the complete pAR set}
\label{tab:mc-hausdorff}
\small
\setlength{\tabcolsep}{7pt}
\begin{tabular}{rrccccc}
\toprule
$r$ & $n$ & Nonempty compact & Unbounded & Median & 75th percentile & 90th percentile \\
\midrule
1 & 100 & 0.9768 & 0.0232 & 11.724 & 21.537 & 40.002 \\
1 & 400 & 1.0000 & 0.0000 & 8.526 & 12.172 & 16.476 \\
1 & 1600 & 1.0000 & 0.0000 & 7.757 & 10.066 & 12.801 \\
\addlinespace[2pt]
4 & 100 & 0.8940 & 0.1060 & 20.693 & 40.853 & 91.828 \\
4 & 400 & 1.0000 & 0.0000 & 13.109 & 18.091 & 23.936 \\
4 & 1600 & 1.0000 & 0.0000 & 11.444 & 14.458 & 17.481 \\
\bottomrule
\end{tabular}
\vspace{2pt}

\begin{minipage}{0.96\textwidth}
\footnotesize
\textit{Notes:} Each row uses 5,000 replications. Monte Carlo standard errors for the two reported probabilities are at most 0.0044. The last three columns summarize $\sqrt n\,d_H$ conditional on a nonempty compact pAR set.
\end{minipage}
\end{table}

At $n=100$, the pAR set is unbounded in 0.0232 of the $r=1$ replications and 0.1060 of the $r=4$ replications. No unbounded set occurs among the 5,000 replications at $n=400$ or $n=1600$ in either design. The scaled Hausdorff quantiles also become smaller. For $r=1$, the median falls from 11.72 to 7.76 and the 90th percentile falls from 40.00 to 12.80. For $r=4$, the corresponding changes are from 20.69 to 11.44 and from 91.83 to 17.48. The results are consistent with the compactness and $O_p(n^{-1/2})$ conclusions in Proposition \ref{prop:hausdorff}. 

The numerical comparisons in Table \ref{tab:mc-computation} examine the singleton, interior, and binding branches under both generic and aligned metrics. Across 360 test problems, the maximum absolute difference between the KKT and spectral statistics is $1.96\times10^{-11}$. In the aligned cases, the maximum difference between the spectral and closed-form statistics is $1.91\times10^{-11}$. The largest difference between the spectral result and an independent generic solver is $1.53\times10^{-9}$, and the largest support-constraint violation is $9.19\times10^{-12}$. No algorithm failure occurs in these cases.

\begin{table}[t]
\centering
\caption{Numerical agreement and a matched timing comparison}
\label{tab:mc-computation}
\small
\setlength{\tabcolsep}{8pt}
\begin{tabular}{lrr}
\toprule
Diagnostic or method & Reported value & Scope \\
\midrule
\multicolumn{3}{l}{\textit{1. Numerical agreement}}\\
Maximum $|\mathrm{KKT}-\mathrm{spectral}|$ & $1.96\times10^{-11}$ & 360 cases \\
Maximum $|\mathrm{spectral}-\mathrm{closed}|$ & $1.91\times10^{-11}$ & 180 aligned cases \\
Maximum $|\mathrm{spectral}-\mathrm{generic}|$ & $1.53\times10^{-9}$ & 360 cases \\
Maximum constraint violation & $9.19\times10^{-12}$ & 360 cases \\
Algorithm failure rate & $0$ & all attempted methods \\
\midrule
\multicolumn{3}{l}{\textit{2. Total time for $r=100$ and 5,001 candidates}}\\
KKT, scalar & $5.9769$ seconds & $1195.15$ $\mu$s per candidate \\
Spectral, scalar & $0.5399$ seconds & $107.96$ $\mu$s per candidate \\
Closed form, scalar & $0.0232$ seconds & $4.64$ $\mu$s per candidate \\
Closed form, vectorized & $0.0030$ seconds & $0.61$ $\mu$s per candidate \\
\bottomrule
\end{tabular}
\vspace{2pt}

\begin{minipage}{0.94\textwidth}
\footnotesize
\textit{Notes:} Timing entries use the aligned case $K=nW$, 30 matched repetitions, and one BLAS thread. The reported times include the norm calculations used by each implementation. Absolute timings are machine dependent.
\end{minipage}
\end{table}

The timing comparison in Table \ref{tab:mc-computation} illustrates the value of separating preprocessing from candidate-specific evaluation. In the largest displayed aligned design, the spectral implementation reduces total time from 5.98 seconds to 0.54 seconds relative to the direct KKT implementation. The scalar closed form requires 0.023 seconds, and vectorizing the closed form reduces the reported time to 0.003 seconds. For a generic metric with $r=5$, the spectral calculation also reduces the reported time per candidate from about 154 microseconds to 71 microseconds when 1,001 candidates are evaluated. Note that these comparisons are specific to the implementation and machine. However, they support the computational advantage of the spectral and closed-form representations for repeated inversion.

Taken together, the simulations are consistent with the four theoretical implications within the maintained design: (i) the singleton experiment tracks the exact size benchmark, (ii) positive-radius profiling is conservative, (iii) local rejection probabilities move toward the boundary limit, and (iv) complete inversion follows the exact tail characterization. The numerical checks also indicate that the KKT, spectral, and closed-form implementations evaluate the same profiled statistic. 

\section{Empirical illustrations}
\label{sec:empirical}

This section analyzes two empirical examples to illustrate the proposed methods in practice. The first revisits the CHR analysis of 401(k) participation and net financial assets. Eligibility is a strong single instrument in that application, and the admissible direct effect has a fixed dollar scale. The second revisits the quarter-of-birth design for returns to schooling (see, \citet{angrist1991does, bound1995problems, bound1996validity}). The quarter-of-birth instruments explain little of the residualized variation in schooling, and the 30-instrument specification provides a useful weak-identification comparison. Thus, the two applications separate sensitivity to the exclusion restriction from sensitivity to the inference method used as a component.

\subsection{Specifications and inferential interpretation}
\label{sec:empirical_implementation}

All specifications include an intercept and the application-specific control variables described below. Let $C$ collect these controls, let $M_C=I-C(C'C)^{-1}C'$, and let $K_C=Z'M_CZ$. The calculations use the control-adjusted statistic developed in Section \ref{sec:uniform}, with denominator degrees of freedom $n-k-r$ and critical value $rF^{-1}_{r,n-k-r}(1-\alpha)$. 

The two applications use different admissible supports. In the 401(k) application, the instrument is binary eligibility and the outcome is measured in thousands of dollars. We use
\[
    \Gamma_{401}(g)=\{\gamma\in\mathbb R:|\gamma|\le g\},
\]
where $g$ is measured in thousands of dollars. This is a fixed coefficient support. In the quarter-of-birth application, a coefficientwise bound would depend on the coding and dimension of the instrument vector. We instead use
\[
    W_{n,C}=\frac{K_C}{n},
    \qquad
    \Gamma_{AK,n}(g)
    =
    \{\gamma:\gamma'W_{n,C}\gamma\le g^2\}.
\]

For comparison, we form the continuous union of Wald intervals over the same support. The exact-exclusion 2SLS standard error is held fixed as $\gamma$ varies. This Wald union with fixed standard error isolates the change from a Wald component to an pAR component. Note that every reported value of $g$ defines a separate pointwise sensitivity analysis, and the curves are not simultaneous confidence bands over $g$. 

Table~\ref{tab:empirical_diagnostics} reports the baseline diagnostics. The 401(k) values are in thousands of dollars. The quarter-of-birth coefficient estimates and the endpoints of exact-exclusion pAR set are multiplied by 100 and are reported in log points. Table~\ref{tab:empirical_sensitivity} reports the intervals and endpoints used in the analysis below. The 401(k) panel reports coefficients in thousands of dollars. The AK panels report $100\beta$ in log points.

\begin{table}[t]
\centering
\caption{Baseline identification diagnostics}
\label{tab:empirical_diagnostics}
\small
\resizebox{\textwidth}{!}{%
\begin{tabular}{lrrrrrrrrc}
\toprule
Specification
& $n$ & $r$ & Partial $R^2$
& First-stage $F$ & Reduced-form $F$
& OLS (SE) & 2SLS (SE) & LIML (SE)
& Exact-exclusion pAR set \\
\midrule
CHR 401(k)
& 9,915 & 1 & 0.557826
& 12,484.32 & 51.63
& 14.57 (1.37) & 13.22 (1.83) & -- 
& $[9.62,16.82]$ \\
AK-3
& 329,509 & 3 & 0.000294
& 32.27 & 9.32
& 7.11 (0.03) & 10.53 (2.01) & --
& $[6.34,15.31]$ \\
AK-30
& 329,509 & 30 & 0.000447
& 4.91 & 1.66
& 7.11 (0.03) & 8.91 (1.61) & 9.29 (1.95)
& $[1.41,17.94]$ \\
\bottomrule
\end{tabular}%
}
\begin{minipage}{0.96\textwidth}
\footnotesize
\textit{Notes:} Partial $R^2$ and the $F$ statistics are calculated after removing the listed controls. Parentheses contain classical homoskedastic standard errors. The LIML standard error is the reported conditional approximation. The exact-exclusion pAR sets invert tests at the 5 percent level. CHR entries are in thousands of dollars. AK point estimates, standard errors, and set endpoints equal $100\beta$ and are in log points.
\end{minipage}
\end{table}

\subsection{401(k) participation and net financial assets}
\label{sec:empirical_401k}
The 401(k) sample analyzed by CHR contains 9,915 observations. The outcome is net financial assets divided by 1,000, the endogenous regressor is 401(k) participation, and the instrument is 401(k) eligibility. The control variables include income category indicators, age and age squared, family size, education indicators, marital status, two earner status, defined benefit pension status, IRA participation, and homeownership. In particular, eligibility has a strong conditional first stage. The partial $R^2$ is 0.558, and the first-stage $F$ statistic is 12,484. 

\begin{table}[t]
\centering
\caption{pAR and Wald sensitivity sets}
\label{tab:empirical_sensitivity}
\small
\begin{tabular}{llcc}
\toprule
Specification & Support radius $g$ & 95\% pAR set & Wald set \\
\midrule
CHR 401(k) & 0     & $[9.62,16.82]$  & $[9.63,16.82]$ \\
           & 1     & $[8.19,18.25]$  & $[8.19,18.25]$ \\
           & 2.5   & $[6.04,20.41]$  & $[6.04,20.40]$ \\
           & 5     & $[2.45,23.99]$  & $[2.45,23.99]$ \\
           & 7.5   & $[-1.15,27.58]$ & $[-1.13,27.58]$ \\
           & 10    & $[-4.74,31.17]$ & $[-4.72,31.16]$ \\
\addlinespace
AK-3       & 0     & $[6.34,15.31]$   & $[6.59,14.46]$ \\
           & 0.005 & $[-4.36,26.93]$  & $[-2.32,23.37]$ \\
           & 0.010 & $[-15.51,38.62]$ & $[-11.22,32.27]$ \\
           & 0.020 & $[-39.24,62.77]$ & $[-29.03,50.08]$ \\
           & 0.030 & $[-63.66,87.31]$ & $[-46.84,67.89]$ \\
\addlinespace
AK-30      & 0     & $[1.41,17.94]$    & $[5.75,12.07]$ \\
           & 0.005 & $[-11.38,31.95]$  & $[-1.47,19.29]$ \\
           & 0.010 & $[-25.09,46.39]$  & $[-8.69,26.51]$ \\
           & 0.020 & $[-54.89,76.73]$  & $[-23.13,40.95]$ \\
           & 0.030 & $[-85.86,107.86]$ & $[-37.57,55.39]$ \\
\bottomrule
\end{tabular}
\begin{minipage}{0.92\textwidth}
\footnotesize
\textit{Notes:} Each row is a separate pointwise 5 percent sensitivity analysis. The Wald comparison holds the exact-exclusion 2SLS SE fixed while taking the continuous support union.
\end{minipage}
\end{table}

Table \ref{tab:empirical_sensitivity} shows that under exact exclusion, the 95 percent pAR set is $[9.62,16.82]$ thousand dollars. The corresponding Wald union is $[9.63,16.82]$. Figure \ref{fig:empirical_401k} demonstrates that the two boundaries remain nearly indistinguishable as the symmetric support expands. At $g=5$, the pAR set is $[2.45,23.99]$. At $g=10$, it is $[-4.74,31.17]$. The same-support Wald endpoints differ by less than 0.02 thousand dollars throughout the reported grid. In this design, replacing the Wald component interval with the pAR component changes little because the first stage is strong.

\begin{figure}[t]
\centering
\caption{Support sensitivity in the 401(k) application}
\includegraphics[width=0.7\textwidth]{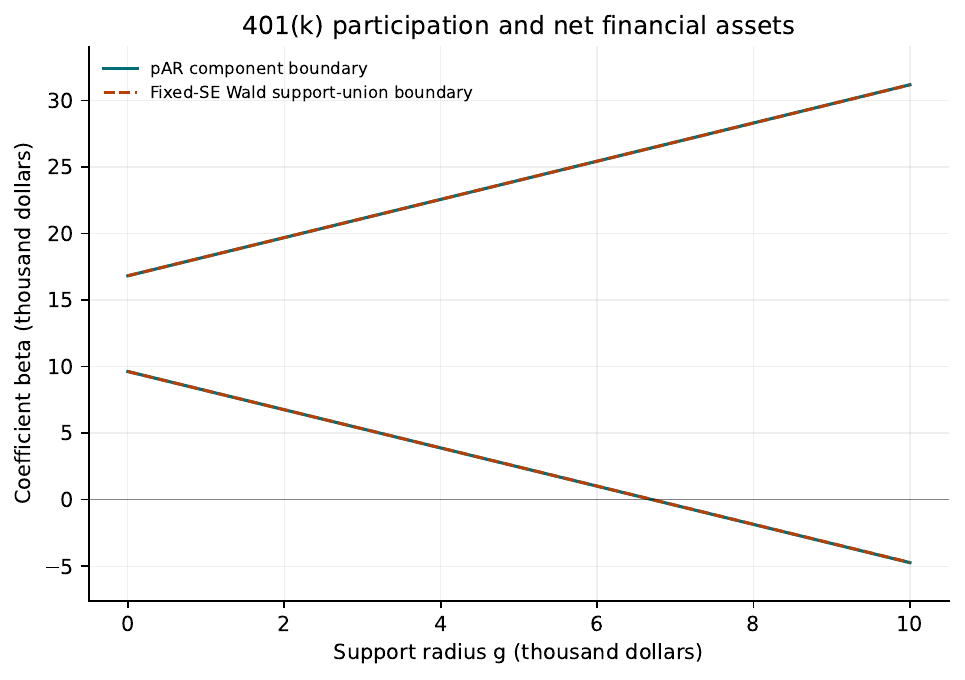}
\label{fig:empirical_401k}
\begin{minipage}{0.92\textwidth}
\footnotesize
\textit{Notes:} The solid curves trace the endpoints of the 95 percent pAR set under $\Gamma_{401}(g)$. The dashed curves trace the continuous Wald support union with fixed SE over the same support.
\end{minipage}
\end{figure}

The zero coefficient remains outside the pAR set at $g=5$, where its p-value is 0.0010, but enters by $g=7.5$, where its p-value is 0.181. The minimum symmetric support radius that retains zero is 6.70 thousand dollars. This value can be interpreted as a sensitivity threshold. It states how large the maintained bound must be before the sample no longer rejects zero at the 5 percent level. 


\subsection{Quarter of birth and returns to schooling}
\label{sec:empirical_qob}

The quarter-of-birth sample analyzed by \cite{angrist1991does} contains 329,509 men born from 1930 through 1939 in the 1980 Census extract. The outcome is log weekly earnings, the endogenous regressor is years of schooling, and the baseline controls are an intercept and nine year-of-birth indicators. The AK-3 specification uses indicators for the first three quarters of birth, with the fourth quarter omitted. The AK-30 specification interacts the three quarter indicators with the ten birth-year indicators. The coefficient is a historical linear return-to-schooling parameter for this sample and specification. 

The two specifications differ sharply in their identification diagnostics provided by Table \ref{tab:empirical_diagnostics}. For AK-3, the first-stage $F$ statistic is 32.27 and the 2SLS estimate is 10.53 log points. The exact-exclusion pAR set is $[6.34,15.31]$, compared with the Wald set $[6.59,14.46]$. For AK-30, the first-stage $F$ statistic falls to 4.91. The 2SLS estimate is 8.91, while the exact-exclusion pAR set widens to $[1.41,17.94]$. The corresponding Wald set is $[5.75,12.07]$. Thus, the pAR and Wald comparisons are modestly different in AK-3 but substantially different in AK-30. The latter difference is consistent with the weak first stage in the 30-instrument specification.

Figure \ref{fig:empirical_qob} traces the sensitivity boundaries using the intervals provided in Table \ref{tab:empirical_sensitivity}. At $g=0.005$, the AK-3 pAR set is $[-4.36,26.93]$ and the AK-30 set is $[-11.38,31.95]$. The respective Wald sets with fixed standard error are $[-2.32,23.37]$ and $[-1.47,19.29]$. The difference between the two procedures increases with $g$. This widening gap is consistent with the pAR inversion retaining weak-identification uncertainty at every admissible direct effect. Every baseline set remains a single bounded interval over the displayed grid.

\begin{figure}[t]
\centering
\caption{Support sensitivity in the quarter-of-birth application}
\includegraphics[width=\textwidth]{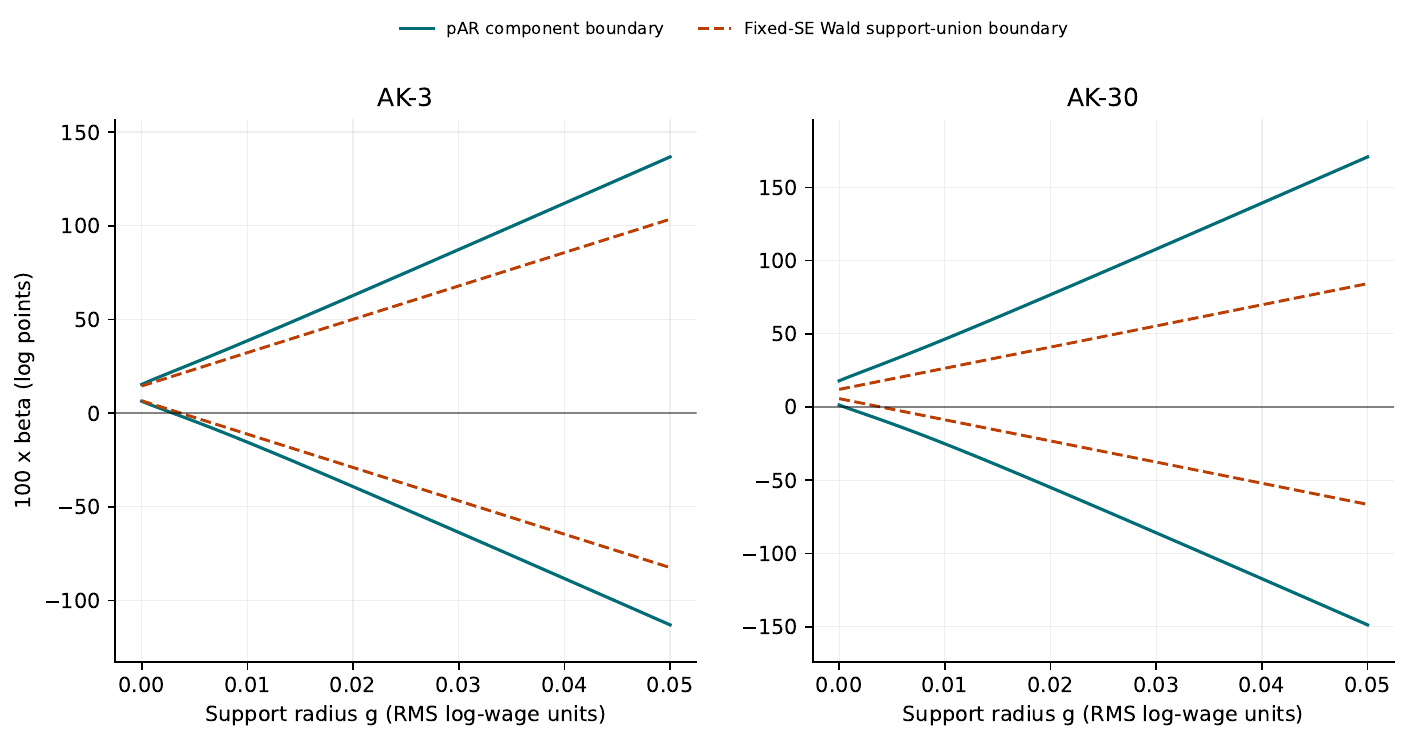}
\label{fig:empirical_qob}
\begin{minipage}{0.94\textwidth}
\footnotesize
\textit{Notes:} The AK support is $\Gamma_{AK,n}(g)=\{\gamma:\gamma'(K_C/n)\gamma\le g^2\}$. Hence, $g$ is a design-conditional sample RMS residualized exclusion departure in log weekly wages. The vertical axis reports $100\beta$ in log points. Solid curves are pAR boundaries, and dashed curves are Wald support union boundaries with fixed SE.
\end{minipage}
\end{figure}

\section{Conclusion}
\label{sec:conclusion}

This paper relates CHR's support restriction under the violation of exclusion restriction to the Anderson--Rubin test. For each candidate treatment effect, pAR profiles over the admissible direct effects of instruments and retains the candidate whenever at least one component test accepts it. Under conditional Gaussian model, the rejection probability does not exceed the nominal level for each fixed member of the identified set, regardless of first-stage strength. If a designated structural coefficient has an admissible direct effect, inversion covers that coefficient. These guarantees are pointwise rather than simultaneous.

The interpretation depends on the exclusion support. A fixed support yields a population membership statement. A support normalized by the realized instrument and control design yields an exact design-conditional statement and has a population root-mean-square interpretation asymptotically. The control-adjusted formulation accommodates standard exogenous covariates. For computation, the KKT and spectral representations reduce the generic binding problem to a scalar search, while sample normalization yields a closed form. These calculations make repeated inversion feasible without a finite grid.

The Monte Carlo results are consistent with the theory within the maintained Gaussian fixed-design setting. Exact-exclusion designs track the nominal benchmark, positive-radius profiling is conservative for compatible candidates, and rejection rises against fixed membership alternatives. Complete inversion follows the predicted tail behavior, and the alternative numerical implementations agree to tolerance. The empirical illustrations clarify when the component test matters. In the strong retirement-saving design based on the 401(k) example, pAR and the Wald union produce nearly the same sensitivity curve. In the weaker quarter-of-birth design, especially under the richer instrument specification, pAR retains substantially more uncertainty even before the exclusion support is enlarged.

The resulting confidence sets remain conditional on the maintained exclusion support. It would be of interest to validate the exclusion restriction, or determine which support is substantively appropriate. Extending the method to heteroskedastic or clustered settings and profiling more powerful weak-instrument tests are natural next steps. Each extension would require a separate validity argument rather than a mechanical substitution of the component statistic.

\newpage
\newpage
\section{Appendix}
\subsection{Algorithms}

\Needspace{0.78\textheight}
\noindent
\begin{minipage}{\linewidth}
\captionsetup{type=algorithm}
\caption{KKT algorithm for computing \(pAR(b;\Gamma)\)}
\label{alg:kkt_pAR_nobreak}

\small
\begin{algorithmic}[1]
\Require Candidate value \(b\), data \((Y,X,Z)\), matrices \(K=Z'Z\succ0\), \(W\succ0\), radius \(g\ge0\)
\Ensure Optimizer \(\gamma^\ast\) and profiled statistic \(pAR(b;\Gamma)\)

\State Compute $\widehat\gamma \gets \widehat\gamma(b) = K^{-1}Z'(Y-Xb)$

\State Compute
\[
    \widehat\sigma_n^2(b)
    \gets
    \frac{(Y-Xb)'M_Z(Y-Xb)}{n-r},
    \qquad
    M_Z=I-ZK^{-1}Z'
\]

\If{\(\widehat\sigma_n^2(b)\le0\)}
    \State \Return undefined statistic
\EndIf

\If{\(g=0\)}
    \State Set \(\gamma^\ast\gets0\)
    \State Set
    \[
        pAR(b;\Gamma(0,W))
        \gets
        \frac{\widehat\gamma'K\widehat\gamma}{\widehat\sigma_n^2(b)}
        =
        AR(b;0)
    \]
    \State \Return \(\gamma^\ast,\ pAR(b;\Gamma(0,W))\)
\EndIf

\If{\(\widehat\gamma'W\widehat\gamma\le g^2\)}
    \State Set \(\gamma^\ast\gets\widehat\gamma\)
    \State Set \(pAR(b;\Gamma)\gets0\)
    \State \Return \(\gamma^\ast,\ pAR(b;\Gamma)\)
\EndIf

\State Find the unique \(\lambda^\ast>0\) solving
\[
    \phi(\lambda)
    =
    \gamma(\lambda)'W\gamma(\lambda)-g^2
    =
    0,
    \qquad
    \gamma(\lambda)
    =
    (K+\lambda W)^{-1}K\widehat\gamma
\]

\State Set
\[
    \gamma^\ast
    \gets
    \gamma(\lambda^\ast)
    =
    (K+\lambda^\ast W)^{-1}K\widehat\gamma
\]

\State Set
\[
    pAR(b;\Gamma)
    \gets
    \frac{
    (\widehat\gamma-\gamma^\ast)'K(\widehat\gamma-\gamma^\ast)
    }
    {\widehat\sigma_n^2(b)}
\]

\State \Return \(\gamma^\ast,\ pAR(b;\Gamma)\)

\end{algorithmic}
\end{minipage}

\Needspace{0.85\textheight}
\noindent
\begin{minipage}{\linewidth}
\captionsetup{type=algorithm}
\caption{Spectral algorithm for computing \(pAR(b;\Gamma)\)}
\label{alg:spectral_par}

\small
\begin{algorithmic}[1]
\Require Candidate value \(b\), data \((Y,X,Z)\), matrices \(K=Z'Z\succ0\), \(W\succ0\), radius \(g\ge0\)
\Ensure Optimizer \(\gamma^\ast\) and profiled statistic \(pAR(b;\Gamma)\)

\State Let \(C=W^{1/2}\)
\State Precompute \(\widetilde K \gets C^{-1}KC^{-1}\)
\State Compute \(\widetilde K=Q\Lambda Q'\), where
\(\Lambda=\operatorname{diag}(\Lambda_1,\ldots,\Lambda_r)\)

\State Compute $\widehat\gamma \gets K^{-1}Z'(Y-Xb)$

\State Compute
\[
    \widehat\sigma_n^2(b)
    \gets
    \frac{(Y-Xb)'M_Z(Y-Xb)}{n-r},
    \qquad
    M_Z=I-ZK^{-1}Z'
\]

\If{\(\widehat\sigma_n^2(b)\le0\)}
    \State \Return undefined statistic
\EndIf

\State Compute \(\widehat u\gets C\widehat\gamma\) and \(\tilde u\gets Q'\widehat u\)

\If{\(g=0\)}
    \State Set \(\gamma^\ast\gets0\)
    \State Set
    \[
        pAR(b;\Gamma(0,W))
        \gets
        \frac{\widehat\gamma'K\widehat\gamma}{\widehat\sigma_n^2(b)}
        =
        AR(b;0)
    \]
    \State \Return \(\gamma^\ast,\ pAR(b;\Gamma(0,W))\)

\ElsIf{\(\|\widehat u\|_2\le g\)}
    \State Set \(\gamma^\ast\gets\widehat\gamma\)
    \State Set \(pAR(b;\Gamma)\gets0\)
    \State \Return \(\gamma^\ast,\ pAR(b;\Gamma)\)

\Else
    \State Find the unique \(\lambda^\ast>0\) solving
    \[
        \phi(\lambda)
        =
        \sum_{j=1}^r
        \left(
            \frac{\Lambda_j}{\Lambda_j+\lambda}
        \right)^2
        \tilde u_j^2
        -
        g^2
        =
        0
    \]

    \State Set
    \[
        \gamma^\ast
        \gets
        C^{-1}Q
        \operatorname{diag}
        \left(
            \frac{\Lambda_j}{\Lambda_j+\lambda^\ast}
        \right)
        Q'\widehat u
    \]

    \State Compute
    \[
        N^\ast
        \gets
        \sum_{j=1}^r
        \Lambda_j
        \left(
            \frac{\lambda^\ast}{\Lambda_j+\lambda^\ast}
        \right)^2
        \tilde u_j^2
    \]

    \State Set \(pAR(b;\Gamma)\gets N^\ast/\widehat\sigma_n^2(b)\)
    \State \Return \(\gamma^\ast,\ pAR(b;\Gamma)\)
\EndIf

\end{algorithmic}
\end{minipage}

\noindent
\begin{minipage}{\linewidth}
\captionsetup{type=algorithm}
\caption{Spectral algorithm for repeated test inversion}
\label{alg:spectral_par_grid_nobreak}

\small
\begin{algorithmic}[1]
\Require Candidate set \(\mathcal B\), data \((Y,X,Z)\), matrices \(K=Z'Z\succ0\), \(W\succ0\), radius \(g\ge0\)
\Ensure Values \(\{pAR(b;\Gamma):b\in\mathcal B\}\)

\State Let \(C\gets W^{1/2}\)
\State Precompute \(\widetilde K\gets C^{-1}KC^{-1}\)
\State Compute $\widetilde K=Q\Lambda Q'$, where $\Lambda=\operatorname{diag}(\Lambda_1,\ldots,\Lambda_r)$.

\State Set \(M_Z\gets I-ZK^{-1}Z'\)

\For{\(b\in\mathcal B\)}
    \State Compute
    \[
        \widehat\gamma
        \gets
        K^{-1}Z'(Y-Xb),
        \qquad
        \widehat\sigma_n^2(b)
        \gets
        \frac{(Y-Xb)'M_Z(Y-Xb)}{n-r}.
    \]

    \If{\(\widehat\sigma_n^2(b)\le0\)}
        \State Set \(pAR(b;\Gamma)\gets\) undefined
        \State \textbf{continue}
    \EndIf

    \State Compute $\widehat u\gets C\widehat\gamma$, $\tilde u\gets Q'\widehat u$.
    
    \If{\(g=0\)}
        \State Set
        \[
            pAR(b;\Gamma(0,W))
            \gets
            \frac{\widehat\gamma'K\widehat\gamma}{\widehat\sigma_n^2(b)}
            =
            AR(b;0).
        \]

    \ElsIf{\(\|\widehat u\|_2\le g\)}
        \State Set $pAR(b;\Gamma)\gets0$.

    \Else
        \State Find the unique \(\lambda^\ast>0\) solving
        \[
            \phi(\lambda)
            =
            \sum_{j=1}^r
            \left(
                \frac{\Lambda_j}{\Lambda_j+\lambda}
            \right)^2
            \tilde u_j^2
            -
            g^2
            =
            0.
        \]

        \State Compute
        \[
            N^\ast
            \gets
            \sum_{j=1}^r
            \Lambda_j
            \left(
                \frac{\lambda^\ast}{\Lambda_j+\lambda^\ast}
            \right)^2
            \tilde u_j^2.
        \]

        \State Set
        \[
            pAR(b;\Gamma)
            \gets
            \frac{N^\ast}{\widehat\sigma_n^2(b)}.
        \]
    \EndIf
\EndFor

\State \Return \(\{pAR(b;\Gamma):b\in\mathcal B\}\)

\end{algorithmic}
\end{minipage}

\newpage
\subsection{Proofs}
 
\begin{proof}[\textbf{Proof of Lemma \ref{lem:par_minimum}}]
By definition,
\[
    pAR(b;\Gamma)
    =
    \inf_{\gamma\in\Gamma}AR(b;\gamma).
\]
Every $\widetilde\gamma\in\Gamma$ is feasible in this infimum. Therefore,
\[
    pAR(b;\Gamma)
    \le
    AR(b;\widetilde\gamma).
\]
Under $H_{0,\mathrm{pop}}^I(b;\Gamma_0)$, take $\widetilde\gamma=\gamma^\dagger(P,b)\in\Gamma_0$.
\end{proof}

\begin{proof}[\textbf{Proof of Proposition \ref{prop:pointwise_validity}}]
Fix $b$ and suppose $\gamma^\dagger(P,b)\in\Gamma_0$. From the reduced-form model,
\[
    Y-Xb-Z\gamma^\dagger(P,b)
    =
    U-bV
    =
    \epsilon_b.
\]
Conditional on $Z$, Assumption \ref{ass:dist} gives $\epsilon_b\sim N(0,\sigma_b^2I_n)$ with $\sigma_b^2>0$. Since $P_Z$ and $M_Z$ are orthogonal projections of ranks $r$ and $n-r$,
\[
    \frac{\epsilon_b'P_Z\epsilon_b}{\sigma_b^2}\mid Z\sim\chi_r^2,
    \qquad
    \frac{\epsilon_b'M_Z\epsilon_b}{\sigma_b^2}\mid Z\sim\chi_{n-r}^2,
\]
and the two quadratic forms are independent. Therefore,
\[
    \frac{AR\{b;\gamma^\dagger(P,b)\}}{r}\mid Z
    \sim F_{r,n-r}.
\]
Lemma \ref{lem:par_minimum} yields
\[
    pAR(b;\Gamma_0)
    \le
    AR\{b;\gamma^\dagger(P,b)\}.
\]
Hence
\[
\{pAR(b;\Gamma_0)>c_{1-\alpha,n}\}
 \subseteq
\{AR\{b;\gamma^\dagger(P,b)\}>c_{1-\alpha,n}\},
\]
and taking conditional probabilities gives 
\[
    \Pr_P\!\left
      \{pAR(b;\Gamma_0)>c_{1-\alpha,n}\}
      \mid Z
    \right)
    \le \alpha.
\]
The coverage statement
\[
    \Pr_P\{b\in \mathcal C_{n,\mathrm{pAR}}(\Gamma_0)\mid Z\}
    \ge 1-\alpha.
\]
is its complement. The unconditional result follows by iterated expectations.
\end{proof}

\begin{proof}[\textbf{Proof of Lemma \ref{lem:distance-denominator}}]
For a fixed candidate \(b\), the numerator of the \(\gamma\)-specific AR statistic is
\[
    \{\ghat(b)-\gamma\}'K\{\ghat(b)-\gamma\}.
\]
Since \(K=n\Qn(Z)\) and \(\Dn(b)\) does not depend on \(\gamma\), profiling over \(\Gammao\) gives \eqref{eq:distance-representation} whenever \(\Dn(b)>0\).

Let \(m=n-r \ge  2\) and let \(H\in\R^{n\times m}\) have orthonormal columns spanning the range of \(M_Z\), so that \(M_Z=HH'\) and $H'H = I_m$.  Define
\[
    \widetilde U=H'U,
    \qquad
    \widetilde V=H'V.
\]
Conditional on \(Z\), the vector \((\widetilde U',\widetilde V')'\) is Gaussian with covariance \(\Sigma_{UV}\otimes I_m\), which is positive definite.  It therefore has a density on \(\R^{2m}\).  Moreover,
\[
    \Dn(b)
    =\frac{1}{m}\|\widetilde U-b\widetilde V\|^2.
\]
If \(\Dn(b)=0\) for some \(b\), then \(\widetilde U\) and \(\widetilde V\) are collinear.  When \(m\ge2\), the set of collinear pairs in \(\R^m\times\R^m\) is a proper algebraic subset of \(\R^{2m}\) and has Lebesgue measure zero.  Hence \eqref{eq:global-denominator-positive} holds.  On this probability-one event, multiplying the ratio inequality by \(\Dn(b)\) gives \eqref{eq:acceptance-inequality} simultaneously for every \(b\in\R\).
\end{proof}

\begin{proof}[\textbf{Proof of Proposition \ref{prop:fixed-consistency}}]
For every fixed \(b \in \mathbb R\),
\[
    \ghat(b)
    =\gdag(P,b)+K^{-1}Z'\epsb.
\]
Conditional on \(Z\),
\[
    K^{-1}Z'\epsb
    \sim N\!\left(0,\sigma_b^2K^{-1}\right).
\]
Because \(K=n\Qn(Z)\) and \(\Qn(Z) \overset{p}{\longrightarrow} W_{ZZ}(P)\succ0\),
\begin{equation}\label{eq:ghat-root-n}
    \ghat(b)-\gdag(P,b)=\Op(n^{-1/2}).
\end{equation}
Also, by conditional Gaussianity given $Z$,
\[
    \frac{(n-r)\Dn(b)}{\sigma_b^2} \mid Z
    \sim\chi^2_{n-r},
\]
and so
\[
    E[D_n(b) \mid Z] = \sigma_b^2, \quad \mathrm{Var} (D_n(b) \mid Z) = \frac{2 \sigma^4(b)}{n-r}.
\]
Therefore,
\begin{equation}
\label{eq:sigma-consistency}
    \Dn(b) \overset{p}{\longrightarrow}\sigma_b^2.
\end{equation}

For fixed $A \succ 0$, the function $f(x,A,\gamma) = (x-\gamma)'A(x-\gamma)$ is jointly continuous in $(x,A,\gamma)$. Because the ellipsoid \(\Gammao\) is compact, the infimum over $\gamma \in \Gammao$ is attained. Thus, the map
\[
    (x,A)\longmapsto d_A^2(x,\Gammao)
\]
is continuous on bounded sets of \(x\) and sets of positive-definite matrices whose eigenvalues are bounded above and away from zero.  Hence \eqref{eq:ghat-root-n} and \(\Qn(Z) \overset{p}{\longrightarrow} W_{ZZ}\) imply
\[
    d_{\Qn(Z)}^2\{\ghat(b),\Gammao\}
    \overset{p}{\longrightarrow}
    d_{W_{ZZ}}^2\{\gdag(P,b),\Gammao\}.
\]
Combining this convergence with Lemma~\ref{lem:distance-denominator} and \eqref{eq:sigma-consistency} proves 
\[
    \frac{1}{n} pAR(b;\Gammao)
    \overset{p}{\longrightarrow}
    \frac{
      d_{W_{ZZ}(P)}^2\{\gdag(P,b),\Gammao\}
    }{
      \sigma_b^2
    }.
\]

If \(b\notin\BI(P;\Gammao)\), then \(\gdag(P,b)\notin\Gammao\).  Closedness of \(\Gammao\) and positive definiteness of \(W_{ZZ}\) imply
\[
    d_{W_{ZZ}(P)}^2\{\gdag(P,b),\Gammao\}>0.
\]
Since \(c_{1-\alpha,n} \longrightarrow q_{r,1-\alpha}<\infty\) where $q_{r,1-\alpha}$ is the $(1-\alpha)$-quantile of $\chi^2_r$, part (i) follows.

If \(\gdag(P,b)\in\operatorname{int}(\Gammao)\), there is an \(\epsilon>0\) such that the Euclidean ball of radius \(\epsilon\) centered at \(\gdag(P,b)\) is contained in \(\Gammao\).  By \eqref{eq:ghat-root-n},
\[
    \Pr_P\!\left\{
       \|\ghat(b)-\gdag(P,b)\|<\epsilon
       \mid Z
    \right\}\to1.
\]
On this event, \(\ghat(b)\in\Gammao\). Thus, the profiled numerator is zero.  This proves part (ii).
\end{proof}

\begin{proof}[\textbf{Proof of Proposition \ref{prop:strict-conservativeness}}]
Define
\[
    N_0(\ghat(b))
    =(\ghat(b) -\gdag(P,b))'K
      (\ghat(b) -\gdag(P,b)),
\]
\[
    N_*(\ghat(b))
    =\inf_{\gamma\in\Gamma}
      (\ghat(b) -\gamma)'K
      (\ghat(b)-\gamma).
\]
Since $\gdag(P,b) \in \Gamma$, $N_*(\widehat \gamma_n(b)) \le N_0(\widehat \gamma_n(b))$ for every realization. Conditional on \(Z\), \(\ghat(b) \) has a nonsingular Gaussian density on \(\R^r\), \(D_n(b)\) has a strictly positive density on \((0,\infty)\), and \(\ghat(b)\) and \(D_n(b)\) are independent.  Moreover,
\[
    \frac{N_0(\ghat(b))}{D_n(b)}
    =\AR(b; \gdag(P,b)),
    \qquad
    \Pr_P\!\left\{
       \frac{N_0}{D_n(b)}>c_{1-\alpha,n}
       \mid Z
    \right\}=\alpha.
\]

If \(\Gamma=\{\gdag(P,b)\}\), then \(N_*=N_0\) because the pAR statistic equals the exact compatible-nuisance AR statistic for every sample realization and the rejection probability is $\alpha$.

Suppose instead that \(\Gamma\) contains \(\gamma_1\ne \gdag(P,b) \).  Let
\[
    \Delta
    =(\gamma_1-\gdag(P,b))'K(\gamma_1- \gdag(P,b))>0,
    \qquad
    d_0=\frac{\Delta}{2c_{1-\alpha,n}}>0.
\]
At \((\ghat(b),D_n(b))=(\gamma_1,d_0)\),
\[
    N_*(\gamma_1)=0,
    \qquad
    c_{1-\alpha,n}d_0=\frac{\Delta}{2},
    \qquad
    N_0(\gamma_1)=\Delta.
\]
The function \(N_0\) is continuous.  The function \(N_*\) is also continuous because it is the minimum of a jointly continuous function over the compact set \(\Gamma\).  Hence, there is an open neighborhood \(\mathcal O\) of \((\gamma_1,d_0)\) on which
\[
    N_*\le c_{1-\alpha,n} D_n(b) <N_0.
\]
The joint density of \((\ghat(b),D_n(b))\) is strictly positive on \(\R^r\times(0,\infty)\). Thus, 
\[
    \Pr_P\{(\ghat(b),D_n(b))\in\mathcal O\mid Z\}>0.
\]
Therefore,
\[
\begin{split}
    \Pr_P\!\left\{\frac{N_*}{D_n(b)}>c_{1-\alpha,n}\mid Z\right\}
    &\le
    \Pr_P\!\left\{\frac{N_0}{D_n(b)}>c_{1-\alpha,n}\mid Z\right\}\\
    &\quad-
    \Pr_P\!\left\{N_*\le c_{1-\alpha,n} D_n(b)<N_0\mid Z\right\}\\
    &<\alpha.
\end{split}
\]
For \(\Gammao(g,W_0)\), \(g=0\) gives the singleton \(\{0_r\}\) because positive definiteness of $W_0$ implies $\gamma'W_0\gamma \le 0 \iff \gamma = 0_r$, while every \(g>0\) gives a set containing more than one point.  The membership-null statement follows.
\end{proof}

\begin{proof}[\textbf{Proof of Proposition \ref{prop:topology-boundedness}}]
Define
\[
    N_n^*(b)
    =\inf_{\gamma\in\Gammao}
       (\dhat-\pihat b-\gamma)'K
       (\dhat-\pihat b-\gamma).
\]
By Lemma \ref{lem:distance-denominator},
\[
    \cset(\Gammao)
    =\{b \in \mathbb R:N_n^*(b)\le c_{1-\alpha,n}\Dn(b)\}.
\]
Equivalently, it is the projection onto the \(b\)-coordinate of
\[
\begin{split}
    \mathcal S_n=\biggl\{(b,\gamma)\in\R\times\R^r:
      &\ \gamma'W_0\gamma\le g^2,\\
      &(\dhat-\pihat b-\gamma)'K
       (\dhat-\pihat b-\gamma)
       \le c_{1-\alpha,n}\Dn(b)
    \biggr\}.
\end{split}
\]
A semialgebraic set is a finite Boolean combination of sets defined by polynomial equalities and inequalities. Here, $\gamma'W_0\gamma-g^2$ is a quadratic polynomial in $\gamma$. Also, $(\dhat-\pihat b-\gamma)'K (\dhat-\pihat b-\gamma) - c_{1-\alpha,n}\Dn(b)$ is a polynomial of degree at most two in $(b,\gamma)$, because $D_n(b)$ is quadratic in $b$. Hence, the set \(\mathcal S_n\) is defined by polynomial inequalities and is semialgebraic. By the Tarski--Seidenberg theorem, its projection is again semialgebraic. 

To establish closedness, let \(b_m\in\cset(\Gammao)\) and \(b_m\to b\).  Choose \(\gamma_m\in\Gammao\) satisfying the acceptance inequality.  Because \(\Gammao\) is compact, $\{\gamma_m\}$ has a convergent subsequence \(\gamma_{m_k}\to\gamma\in\Gammao\).  Passing to the limit in the continuous inequality gives \(b\in\cset(\Gammao)\).  Every closed semialgebraic subset of \(\R\) has finitely many connected components and has the form stated in part (i).

For the tail limit, compactness of \(\Gammao\) gives
\[
    \frac{N_n^*(b)}{b^2}
    =\inf_{\gamma\in\Gammao}
      \left\{
       \frac{\dhat-\gamma}{b}-\pihat
      \right\}'K
      \left\{
       \frac{\dhat-\gamma}{b}-\pihat
      \right\} 
    \longrightarrow
    \pihat'K\pihat 
\]
as \(|b|\to\infty\), because 
$$
    \sup_{\gamma \in \Gamma_0} \bigg\|\frac{\widehat\delta_n - \gamma}{b} \bigg\|
    \longrightarrow 0.
$$
Also,
\[
    \frac{\Dn(b)}{b^2}
    \longrightarrow
    \frac{X'M_ZX}{n-r}=s_{X,n}^2>0.
\]
Therefore,
\[
    \lim_{|b|\to\infty} pAR(b;\Gammao)=L_n.
\]
Suppose first that \(L_n>c_{1-\alpha,n}\). By convergence of the statistic to $L_n$ as $|b| \to \infty$, there exist $M < \infty$ such that
$$
    |b| > M \quad \implies \quad pAR(b;\Gamma_0) > c_{1-\alpha,n}.
$$
Thus, no accepted value lies outside $[-M,M]$. Since the confidence set is closed, it is a closed subset of a compact interval and hence compact and bounded. If \(L_n<c_{1-\alpha,n}\), the same convergence gives an $M$ such that
$$
    |b| > M \quad \implies \quad pAR(b;\Gamma_0) < c_{1-\alpha,n}.
$$
Thus, both rays $(-\infty, -M]$ and $[M, \infty)$ lie in the acceptance set.
\end{proof}

\begin{proof}[\textbf{Proof of Corollary \ref{cor:tail-first-stage}}]
Under Assumption \ref{ass:dist},
\[
    \pihat=\pi(P)+K^{-1}Z'V.
\]
Conditional on \(Z\),
\[
    \frac{\pihat'K\pihat}{\sigma_V^2}
    \sim\chi_r^2(\lambda_{\pi,n}),
    \qquad
    \lambda_{\pi,n}
    =\frac{\pi(P)'K\pi(P)}{\sigma_V^2}.
\]
Because \(M_ZX=M_ZV\),
\[
    \frac{(n-r)s_{X,n}^2}{\sigma_V^2}
    =\frac{V'M_ZV}{\sigma_V^2}
    \sim\chi_{n-r}^2.
\]
The two quadratic forms are independent because they depend on the orthogonal Gaussian projections \(P_ZV\) and \(M_ZV\).  Consequently,
$$
    \frac{L_n}{r} \sim F_{r,n-r}(\lambda_{\pi,n}).
$$

If \(\pi(P)\ne0_r\), then \(\pihat\to_p\pi(P)\), \(s_{X,n}^2\to_p\sigma_V^2\), and
\[
    \frac{L_n}{n}
    =\frac{\pihat'\Qn(Z)\pihat}{s_{X,n}^2}
    \overset{p}{\longrightarrow}
    \frac{\pi(P)'W_{ZZ}\pi(P)}{\sigma_V^2}>0.
\]
Since \(c_{1-\alpha,n}=O(1)\), Proposition~\ref{prop:topology-boundedness} implies bounded inversion with probability approaching one.  If \(\pi(P)=0_r\), then $L_n/r \sim F_{r,n-r}$.  Continuity of this distribution gives
\[
    \Pr_P\{L_n<c_{1-\alpha,n}\mid Z\}=1-\alpha,
\]
and the equality event has probability zero.  Proposition~\ref{prop:topology-boundedness} completes the proof.
\end{proof}

\begin{proof}[\textbf{Proof of Proposition \ref{prop:hausdorff}}]
Define the sample plug-in identified set
\[
    \widehat B_{I,n}
    =\{b\in\R:\ghat(b)\in\Gammao\}
    =\{b:(\dhat-\pihat b)'W_0(\dhat-\pihat b)\le g^2\}.
\]
Conditional normality and \(\Qn(Z) \rightarrow W_{ZZ}\succ0\) give
\[
    \dhat-\delta(P)=\Op(n^{-1/2}),
    \qquad
    \pihat-\pi(P)=\Op(n^{-1/2}).
\]
Set
\[
    \widehat a_{W_0}=\pihat'W_0\pihat,
    \qquad
    \widehat b_{W_0}=\pihat'W_0\dhat,
    \qquad
    \widehat d_{W_0}=\dhat'W_0\dhat.
\]
Because \(a_{W_0}>0\) and \(g^2-g_*^2>0\), with probability approaching one \(\widehat a_{W_0}>0\) and
\[
    g^2-
    \left(
       \widehat d_{W_0}
       -\frac{\widehat b_{W_0}^2}{\widehat a_{W_0}}
    \right)>0.
\]
On this event, \(\widehat B_{I,n}=[\widehat b_L,\widehat b_U]\), where
\[
    \widehat b_{L,U}
    =\frac{\widehat b_{W_0}}{\widehat a_{W_0}}
     \mp
     \frac{
       \sqrt{
          g^2-\widehat d_{W_0}
          +\widehat b_{W_0}^2/\widehat a_{W_0}
       }
     }{
       \sqrt{\widehat a_{W_0}}
     }.
\]
The endpoint map is continuously differentiable function of $(\widehat a_{W_0}, \widehat b_{W_0}, \widehat d_{W_0})$ in a neighborhood of the population values whenever $\widehat a_{W_0 } > 0$ and the expression under the square root is strictly positive. The vector of sample coefficients is root-$n$ consistent. Hence, a first-order Taylor expansion of each endpoint around its population value gives $\widehat b_L - b_L = O_p(n^{-1/2})$ and $\widehat b_U - b_U = O_p(n^{-1/2})$. For intervals, the Hausdorff distance equals the maximum absolute endpoint error. Thus, the delta method gives
\[
    d_H\{\widehat B_{I,n},\BI(P;\Gammao)\}
    =\Op(n^{-1/2}).
\]
If \(b\in\widehat B_{I,n}\), then $\widehat \gamma_n(b) \in \Gamma_0$. The profiling problem can choose $\gamma = \widehat \gamma_n(b)$ which yields a zero profiled numerator. Thus,
\begin{equation}
    \label{eq:plugin-subset}
    \widehat B_{I,n}\subseteq\cset(\Gammao).
\end{equation}
Hence,
\[
    d\{b, \mathcal C_{n,\rm pAR}(\Gamma_0) \} \le d\{ b, \widehat B_{I,n} \}.
\]
It follows that
\begin{equation}\label{eq:inner-Hausdorff}
    \sup_{b\in\BI(P;\Gammao)}
    d\{b,\cset(\Gammao)\}
    =\Op(n^{-1/2}).
\end{equation}

We next establish compact containment.  Let
\[
    R_\Gamma=\sup_{\gamma\in\Gammao}\|\gamma\|<\infty.
\]
The model and the Gaussian projection identities imply
\[
    \dhat \overset{p}{\longrightarrow} \delta(P),
    \qquad
    \pihat \overset{p}{\longrightarrow} \pi(P),
\]
\[
    \frac{\|M_ZY\|^2}{n-r}
    =\frac{U'M_ZU}{n-r}
    \overset{p}{\longrightarrow} \sigma_U^2,
\]
\[
    \frac{\|M_ZX\|^2}{n-r}
    =\frac{V'M_ZV}{n-r}
    \overset{p}{\longrightarrow} \sigma_V^2,
\]
and \(c_{1-\alpha,n} \to q_{r,1-\alpha}\).  Choose fixed constants
\[
    \underline q>0,
    \quad p_0>0,
    \quad D_0>0,
    \quad A_0>0,
    \quad B_0>0,
    \quad C_0>0
\]
such that the event
\[
\begin{split}
    \mathcal E_n=\biggl\{&
       \lambda_{\min}\{\Qn(Z)\}\ge\underline q,
       \ \|\pihat\|\ge p_0,
       \ \|\dhat\|\le D_0,\\
       &\quad \frac{\|M_ZY\|}{\sqrt{n-r}}\le A_0,
       \ \frac{\|M_ZX\|}{\sqrt{n-r}}\le B_0,
       \ c_{1-\alpha,n}\le C_0
    \biggr\}
\end{split}
\]
satisfies \(\Pr_P(\mathcal E_n\mid Z)\to1\).

On \(\mathcal E_n\),
\[
\begin{split}
    d_{\Qn(Z)}\{\ghat(b),\Gammao\}
    &\ge
      \sqrt{\underline q}\,
      d\{\dhat-\pihat b,\Gammao\}\\
    &\ge
      \sqrt{\underline q}\,
      \bigl(p_0|b|-D_0-R_\Gamma\bigr)_+,
\end{split}
\]
whereas by the triangular inequality,
\[
    \sighat(b)
    \le A_0+B_0|b|.
\]
If \(b\in\cset(\Gammao)\), then
\[
    \sqrt n\,
    d_{\Qn(Z)}\{\ghat(b),\Gammao\}
    \le
    \sqrt{c_{1-\alpha,n}}\,\sighat(b).
\]
On $\mathcal E_n$, any accepted $b$ with $p_0 |b| - D_0 - R_\Gamma > 0$ must satisfy
\[
    \sqrt{n \underline q} \{ p_0 |b| - D_0 - R_\Gamma \} \le \sqrt{C_0} (A_0 + B_0 |b|).
\]
Rearranging the coefficients of $|b|$ gives
\[
    \{ \sqrt{n \underline q} p_0 - \sqrt{C_0}B_0 \} |b| \le \sqrt{n \underline q} (D_0 + R_\Gamma) + \sqrt{C_0}A_0.
\]
For all sufficiently large \(n\), this inequality and the preceding bounds imply \(|b|\le M\) for a fixed finite \(M\) on $\mathcal E_n$.  Since $\Pr_P (\mathcal E_n \mid Z) \to 1$,
\begin{equation}\label{eq:compact-containment}
    \Pr_P\!\left\{
       \cset(\Gammao)\subset[-M,M]
       \mid Z
    \right\}\longrightarrow1.
\end{equation}

On \([-M,M]\),
\[
    \sup_{|b|\le M}
    \|\ghat(b)-\gdag(P,b)\|
    \le
    \|\dhat-\delta(P)\|
    +M\|\pihat-\pi(P)\|
    =\Op(n^{-1/2}),
\]
and \(\sup_{|b|\le M}\sighat(b)=\Op(1)\). Therefore, the acceptance inequality gives
\[
    d_{\Qn(Z)}\{\ghat(b),\Gammao\} \le \frac{\sqrt{c_{1-\alpha,n}}\widehat \sigma_n(b)}{\sqrt{n}}
    =\Op(n^{-1/2})
\]
uniformly over accepted \(b\in[-M,M]\).  The lower eigenvalue bound and the preceding uniform approximation imply
\begin{equation}\label{eq:population-line-distance}
    d\{\gdag(P,b),\Gammao\}
    =\Op(n^{-1/2})
\end{equation}
uniformly over accepted $b$.

It remains to convert nuisance-space distance into scalar-parameter distance.  Define
\[
    \rho(b)=d\{\gdag(P,b),\Gammao\},
    \qquad
    \varphi(b)
    =\gdag(P,b)'W_0\gdag(P,b)-g^2.
\]
$\rho(b)$ is zero when $\gamma^\dagger (P,b) \in \Gamma_0$, so the zero set of \(\rho(b)\) is \([b_L,b_U]\).  At either endpoint \(b_e\), let \(\gamma_e=\gdag(P,b_e)\) and \(a_e=W_0\gamma_e\).  Since \(g>g_*\),
\[
    \varphi'(b_e)=-2a_e'\pi(P)\ne0.
\]
For \(b\) just outside the interval,
\[
    \varphi(b)
    \ge c\,d\{b,[b_L,b_U]\}
\]
for some \(c>0\).  For every \(\gamma\in\Gammao\), boundedness on \([-M,M]\) gives
\[
\begin{split}
    \varphi(b)
    &\le
      \gdag(P,b)'W_0\gdag(P,b)-\gamma'W_0\gamma\\
    &=\{\gdag(P,b)-\gamma\}'W_0
      \{\gdag(P,b)+\gamma\}\\
    &\le C\|\gdag(P,b)-\gamma\|
\end{split}
\]
for a finite \(C\).  Taking the infimum over \(\gamma\in\Gammao\) yields a linear error bound near the endpoints.  On the compact portion of \([-M,M]\setminus[b_L,b_U]\) away from the endpoints, continuity gives the same bound with a possibly smaller constant.  Thus, there is \(\kappa_M>0\) such that
\begin{equation}\label{eq:error-bound}
    \rho(b)
    \ge
    \kappa_Md\{b,\BI(P;\Gammao)\}
    \qquad\text{for all }|b|\le M.
\end{equation}
Combining \eqref{eq:population-line-distance} and \eqref{eq:error-bound} gives
\begin{equation}\label{eq:outer-Hausdorff}
    \sup_{b\in\cset(\Gammao)}
    d\{b,\BI(P;\Gammao)\}
    =\Op(n^{-1/2}).
\end{equation}
The plug-in interval is nonempty with probability approaching one. Then  \eqref{eq:plugin-subset} makes \(\cset(\Gammao)\) nonempty.  \eqref{eq:inner-Hausdorff}, \eqref{eq:compact-containment},  and \eqref{eq:outer-Hausdorff} prove
\[
    d_H^{\mathrm{ex}}\!\left(
      \cset(\Gammao),\BI(P;\Gammao)
    \right)
    =\Op(n^{-1/2}).
\]
\end{proof}

\begin{proof}[\textbf{Proof of Lemma \ref{lem:tangent-distance}}]
Because \(W_0\succ0\), every \(v\in\mathcal T_n\) satisfies \(a_0'v\le0\). Hence, \(\mathcal T_n\subseteq\mathcal T\) and the approximating sets never extend outside the tangent half-space.  Conversely, fix \(R<\infty\). We want to show that there is a constant \(C_R<\infty\) such that, for every \(v\in\mathcal T\cap\{\|v\|\le R\}\), the inward shift
\[
    v_n=v-\frac{C_R}{\sqrt n}a_0
\]
belongs to \(\mathcal T_n\) for all sufficiently large \(n\).  Substituting $v_n = v -C_r a_0 / \sqrt{n}$ into the linear term gives
\[
    2a_0'v_n = 2a_0'v - \frac{2C_R}{\sqrt n} \| a_0 \|^2.
\]
Because $v \in \mathcal T$, $a_0'v \le 0$. Thus, the inward shift yields a negative margin of order $n^{-1/2}$. On the ball $\| v\| \le R$, the shifted vectors $v_n$ remain in a slightly larger fixed ball for large $n$. Hence, for a fixed constant $B_R$,
\[
    v_n' W_0 v_n \le B_R.
\]
Choosing $C_R$ so that
\[
    2C_R \| a_0 \|^2 > B_R
\]
ensures
\[
    2a_0'v_n + n^{-1/2}v_n' W_0 v_n \le 0.
\]
Therefore, $v_n \in \mathcal T_n$ and the quadratic term is uniformly bounded by a constant times \(n^{-1/2}\) on the relevant compact set. This shows that every bounded point of the tangent half-space is within order $n^{-1/2}$ of the curved set. Therefore,
\[
    \sup_{v\in\mathcal T\cap\{\|v\|\le R\}}
    d(v,\mathcal T_n)
    =O(n^{-1/2}).
\]
The feasible sets converge locally in Hausdorff distance on every bounded ball. Together with \(\mathcal T_n\subseteq\mathcal T\), this gives the local inner and outer approximations needed for convergence of the minimized value functions.

Let \(v_n^*\) minimize 
\[
    \inf_{v \in \mathcal T_n} (t_n - v)' W_n(Z) (t_n - v).
\]
Since \(0\in\mathcal T_n\),
\[
    (t_n-v_n^*)'\Qn(Z)(t_n-v_n^*)
    \le t_n'\Qn(Z)t_n.
\]
The eigenvalues of \(\Qn(Z)\) are bounded above and away from zero. Let $\underline q$ and $\overline q$ be uniform eigenvalue bounds. Then,
\[
    \underline q \|t_n - v_n^* \|^2 \le (t_n - v_n^*)' W_n(Z) (t_n - v_n^*) \le \overline q \| t_n \|^2.
\]
Thus, 
\[
    \| t_n - v_n^* \| \le \sqrt{ \overline q / \underline q }  \|t_n\|.
\]
Since \(t_n=O_p(1)\), the difference $t_n - v_n^*$ is $O_p(1)$. The triangular inequality gives $\| v_n^*\| \le \| t_n \| +  \| t_n - v_n^* \| =  O_p(1)$.  

For a fixed large radius $R < \infty$, stochastic boundedness of $v_n^*$ and $t_n$ makes the event that all relevant points lie in the ball of radius $R$ arbitrarily likely. On that ball, $\mathcal T_n$ converges to $\mathcal T$ in local Hausdorff distance, $W_n(Z) \to W_{ZZ}(P)$, and the map $(t,v,A) \mapsto (t-v)' A (t-v) $ is uniformly continuous. These facts imply that the value function
\[
    \psi_n(t) = \inf_{v \in \mathcal T_n} (t-v)' W_n(Z) (t-v)
\]
converge uniformly on compact sets to
\[
    \psi(t) =  \inf_{v \in \mathcal T} (t-v)' W_{ZZ}(P) (t-v).
\]
Furthermore, the function $\psi$ is continuous. Therefore,
\[
    \psi_n(t_n) - \psi(t_n) \overset{p}{\longrightarrow} 0
\]
and
\[
    \psi(t_n) \Rightarrow \psi(t)
\]
by the continuous mapping theorem. Combining the two statements proves convergence of the minimum values.

\end{proof}

\begin{proof}[\textbf{Proof of Proposition \ref{prop:boundary-local-power}}]
Because \(b_0\) is an endpoint of the nondegenerate interval, $\gamma_0'W_0\gamma_0=g^2$. The derivative of
\[
    b\mapsto\gdag(P,b)'W_0\gdag(P,b)
\]
at \(b_0\) is \(-2a_0'\pi(P)\).  The intersection is transverse when \(g>g_*\), so \(a_0'\pi(P)\ne0\).

For \(b_n=b_0+h/\sqrt n\),
\[
    \gdag(P,b_n)
    =\gamma_0-\frac{\pi(P)h}{\sqrt n}.
\]
For each $n$,
\[
    \widehat \gamma_n(b_n) - \gamma^\dagger(P,b_n) = K^{-1}Z' \epsilon_{b_n}.
\]
Conditional on $Z$, this vector is Gaussian with covariance
\[
    \sigma^2_{b_n}K^{-1} = \frac{\sigma^2_{b_n}}{n} w_n(Z)^{-1}.
\]
Therefore,
\[
    \sqrt{n} \{\ghat(b_n)-\gdag(P,b_n)\} \mid Z \sim N(0, \sigma^2_{b_n} W_n(Z)^{-1} ).
\]
Since $b_n \to b_0$, the variance function is continuous and $\sigma^2_{b_n} \to \sigma^2_0$. Also, $W_n(Z)^{-1} \to W_{ZZ}(P)^{-1}$. Convergence of the Gaussian covariance matrices gives the weak limit 
\[
    \sqrt n\{\ghat(b_n)-\gdag(P,b_n)\} \mid Z 
    \Rightarrow
    \xi,
    \qquad
    \xi\sim N(0,\sigma_0^2W_{ZZ}^{-1}).
\]
Hence, adding and subtractracting $\gamma^\dagger(P,b_n)$ gives
\[
    t_n:=\sqrt n\{\ghat(b_n)-\gamma_0\} \mid Z
    \Rightarrow
    t:=\xi-\pi(P)h.
\]

The rescaled feasible set is
\[
    \sqrt n(\Gammao-\gamma_0)
    =\left\{
       v:
       2a_0'v+n^{-1/2}v'W_0v\le0
     \right\}
    =\mathcal T_n.
\]
Therefore,
\[
    n\,d_{\Qn(Z)}^2\{\ghat(b_n),\Gammao\}
    =\inf_{v\in\mathcal T_n}
      (t_n-v)'\Qn(Z)(t_n-v).
\]
Because $t_n \Rightarrow t$, $W_n(Z) \to W_{ZZ}(P) \succ 0$ by the design assumption, and $\mathcal T_n$ is the curved local set from Lemma \ref{lem:tangent-distance}, 
\begin{equation}
    \label{lem:replace_curved_set_tangent_half}
    n\,d_{\Qn(Z)}^2\{\ghat(b_n),\Gammao\}
    \Rightarrow
    d_{W_{ZZ}}^2(t,\mathcal T),
\end{equation}
where \(\mathcal T=\{v:a_0'v\le0\}\).

We want to show that for this half-space, the \(W_{ZZ}\)-metric projection gives
\begin{equation}
    \label{lem:half-space-formula}
    d_{W_{ZZ}(P)}^2(x,\mathcal T)
    =\frac{\{a_0'x\}_+^2}{a_0'W_{ZZ}^{-1}a_0}
    =\frac{\{a_0'x\}_+^2}{v_0}.
\end{equation}
If $a_0'x \le 0$, then $x \in \mathcal T$. Hence, the distance is zero and $\{ a_0'x \}_+ = 0$. Suppose $a_0'x > 0$. The closest point lies on the boundary $a_0'v = 0$. Consider the minimization problem
\[
    \min_v (x-v)' W_{ZZ}(P) (x-v) \quad \text{subject to} \quad a_0'v = 0.
\]
The Lagrangian first-order condition is
\[
    -2 W_{ZZ}(P)(x-v) + 2 \lambda a_0 = 0.
\]
Thus,
\[
    v = x - \lambda W_{ZZ}(P)^{-1} a_0.
\]
Imposing $a_0'v = 0$ gives
\[
    \lambda = \frac{a_0'x}{a_0' W_{ZZ}(P)^{-1} a_0 }.
\]
Substitution into the quadratic objective yields
\[
    \frac{(a_0'x)^2}{a_0' W_{ZZ}(P)^{-1} a_0}.
\]
Combining the feasible and infeasible cases gives the positive part.

Because $\text{Var}(a_0'\xi) = \sigma^2_0  v_0$, 
\[
    G = \frac{a_0' \xi}{\sigma_0 \sqrt{v_0}}
\]
is standard normal. Since $t = \xi - \pi(P)h$,
\[
    \frac{a_0't}{\sigma_0 \sqrt{v_0}} = \frac{a_0'\xi}{\sigma_0 \sqrt{v_0}} - \frac{h a_0' \pi(P)}{\sigma_0 \sqrt{v_0}} = G - \kappa(h).
\]
Then, (\ref{lem:half-space-formula}) gives
\begin{equation}    \label{lem:standardize_gaussian_component}
    \frac{d^2_{W_{ZZ}(P)}(t, \mathcal T)}{\sigma^2_0} = \{ G - \kappa(h) \}^2_+.
\end{equation}
Combining (\ref{lem:replace_curved_set_tangent_half}) and (\ref{lem:standardize_gaussian_component}) gives
\[
    \frac{
      n\,d_{\Qn(Z)}^2\{\ghat(b_n),\Gammao\}
    }{
      \sigma_0^2
    }
    \Rightarrow
    \{G-\kappa(h)\}_+^2.
\]
Also, \(\Dn(b_n) \overset{p}{\longrightarrow}\sigma_0^2\).  Lemma \ref{lem:distance-denominator} and Slutsky's theorem prove 
\begin{equation}
\label{lem:boundary-local-limit}
    pAR(b_n;\Gammao)
    \Rightarrow
    \{G-\kappa(h)\}_+^2,
    \qquad G\sim N(0,1),
\end{equation}

The critical value satisfies \(c_{1-\alpha,n}\to q_{r,1-\alpha}>0\).  The limiting distribution in \eqref{lem:boundary-local-limit} has no atom at this positive threshold. Thus, convergence of rejection probabilities follows.  Moreover,
\[
    \{G-\kappa(h)\}_+^2>q_{r,1-\alpha}
    \quad\Longleftrightarrow\quad
    G>\sqrt{q_{r,1-\alpha}}+\kappa(h),
\]
which proves 
\[
    \lim_{n\to\infty}
    \Pr_P\!\left\{
       pAR(b_n;\Gammao)>c_{1-\alpha,n}
       \mid Z
    \right\}
    =1-\Phi\!\left(
       \sqrt{q_{r,1-\alpha}}+\kappa(h)
     \right),
\]
The direction \(-\pi(P)h\) is outward precisely when
\[
    a_0'\{-\pi(P)h\}>0
    \quad\Longleftrightarrow\quad
    h\,a_0'\pi(P)<0.
\]
At \(h=0\), the limiting rejection probability is \(1-\Phi\{\sqrt{q_{r,1-\alpha}}\}\).  For integer $r \ge 1$, a $\chi^2_r$ random variable can be represented as a $\chi^2_1$ plus an independent nonnegative $\chi^2_{r-1}$. Hence, \(\chi_r^2\) stochastically dominates \(\chi_1^2\) and its $(1-\alpha)$-quantile is at least as large. A $\chi^2_1$ random variable is the square of a standard normal random variable. Thus,
\[
    \Pr(\chi^2_1 \le z^2_{1-\alpha/2}) = \Pr(|G| \le z_{1-a/2}) = 1- \alpha.
\]
Thus, $q_{1,1-\alpha} = z^2_{1-\alpha/2}$. Then, monotonicity of $1 - \Phi(x)$ gives
\[
    1-\Phi\{\sqrt{q_{r,1-\alpha}}\}
    \le
    1 - \Phi(z_{1-\alpha/2}) = \frac{\alpha}{2}.
\]
Equality holds for $r=1$ and the inequality is strict for $r > 1$.
\end{proof}

\begin{proof}[\textbf{Proof of Proposition \ref{prop:kkt}}]
If $g=0$, positive definiteness of $W_n$ gives $\Gamma_n(0,W_n)=\{0\}$. The unique feasible point is therefore $\gamma^\ast=0$, and
\[
    N_n(b;\gamma^\ast)=\widehat\gamma'K\widehat\gamma,
    \qquad
    pAR\{b;\Gamma_n(0,W_n)\}
    =\frac{\widehat\gamma'K\widehat\gamma}{\widehat\sigma_n^2(b)}
    =AR(b;0).
\]
If $\widehat\gamma\ne0$, no finite multiplier yields $\gamma(\lambda)=0$. For every finite $\lambda\ge0$, the matrix $(K+\lambda W_n)^{-1}K$ is nonsingular, so
\[
    \gamma(\lambda)
    =(K+\lambda W_n)^{-1}K\widehat\gamma
    \ne0.
\]
Hence $\phi(\lambda)>0$ when $g=0$, and the boundary solution is reached only as $\lambda\to\infty$.

Suppose now that $g>0$. The zero vector is strictly feasible, so Slater's condition holds and the KKT conditions are necessary and sufficient. Stationarity gives
\begin{equation}
    2K(\gamma-\widehat\gamma)+2\lambda W_n\gamma=0
    \quad\Longrightarrow\quad
    \gamma(\lambda)=(K+\lambda W_n)^{-1}K\widehat\gamma.
    \label{gamma_lambda}
\end{equation}
Feasibility and complementary slackness require
\[
    \gamma(\lambda)'W_n\gamma(\lambda)\le g^2,
    \qquad
    \lambda\ge0,
    \qquad
    \lambda\{\gamma(\lambda)'W_n\gamma(\lambda)-g^2\}=0.
\]

If $\widehat\gamma\in\Gamma_n(g,W_n)$, the unconstrained minimizer is feasible. Thus, $\gamma^\ast=\widehat\gamma$, $\lambda^\ast=0$, and $pAR\{b;\Gamma_n(g,W_n)\}=0$.

If $\widehat\gamma\notin\Gamma_n(g,W_n)$, the constraint binds and $\lambda^\ast>0$ solves
\begin{equation}
    \phi(\lambda)
    =\gamma(\lambda)'W_n\gamma(\lambda)-g^2
    =0,
    \qquad \lambda>0.
    \label{phi_lambda}
\end{equation}
The endpoint values satisfy $\phi(0)=\widehat\gamma'W_n\widehat\gamma-g^2>0$ and $\lim_{\lambda\to\infty}\phi(\lambda)=-g^2<0$. Moreover,
\[
    \phi'(\lambda)
    =-2\gamma(\lambda)'W_n
      (K+\lambda W_n)^{-1}W_n\gamma(\lambda)<0
\]
for every finite $\lambda$ in the binding case. Continuity and strict monotonicity therefore imply that the root exists and is unique.

Finally,
\[
    \widehat\gamma-\gamma(\lambda)
    =\lambda(K+\lambda W_n)^{-1}W_n\widehat\gamma.
\]
Substitution of $\lambda=\lambda^\ast$ gives
\[
    N_n(b;\gamma^\ast)
    =\lambda^{\ast2}\widehat\gamma'W_n
      (K+\lambda^\ast W_n)^{-1}K
      (K+\lambda^\ast W_n)^{-1}W_n\widehat\gamma.
\]
Division by $\widehat\sigma_n^2(b)$ yields the stated expression for $pAR\{b;\Gamma_n(g,W_n)\}$.
\end{proof}

\begin{proof}[\textbf{Proof of Corollary \ref{cor:dual}}]
Define the Lagrangian
\[
    \mathcal L(\gamma,\lambda)
    =N_n(b;\gamma)+\lambda(\gamma'W_n\gamma-g^2),
    \qquad \lambda\ge0.
\]
For a finite $\lambda\ge0$,
\[
    N_n(b;\gamma)+\lambda\gamma'W_n\gamma
    =\gamma'(K+\lambda W_n)\gamma
     -2\gamma'K\widehat\gamma
     +\widehat\gamma'K\widehat\gamma.
\]
Let $A=K+\lambda W_n$ and $c=K\widehat\gamma$. Since $A\succ0$,
\[
    \inf_{\gamma\in\mathbb R^r}
    \{\gamma'A\gamma-2c'\gamma\}
    =-c'A^{-1}c.
\]
It follows that
\[
    q(\lambda)
    =\inf_{\gamma\in\mathbb R^r}\mathcal L(\gamma,\lambda)
    =\widehat\gamma'K\widehat\gamma
     -\widehat\gamma'K(K+\lambda W_n)^{-1}K\widehat\gamma
     -\lambda g^2.
\]

For any feasible $\gamma$,
\[
    q(\lambda)
    =\inf_{\widetilde\gamma\in\mathbb R^r}
      \mathcal L(\widetilde\gamma,\lambda)
    \le \mathcal L(\gamma,\lambda)
    \le N_n(b;\gamma),
\]
because $\lambda\ge0$ and $\gamma'W_n\gamma-g^2\le0$. Therefore, $q(\lambda)\le V_n(b,g)$ for every $\lambda\ge0$.

If $g>0$, the zero vector is strictly feasible and Slater's condition gives strong duality. Hence
\[
    V_n(b,g)=q(\lambda^\ast),
\]
where $\lambda^\ast$ is the finite optimal multiplier from Proposition \ref{prop:kkt}. Differentiation, or the envelope theorem, gives
\[
    q'(\lambda)
    =\gamma(\lambda)'W_n\gamma(\lambda)-g^2.
\]

If $g=0$, the primal feasible set is $\{0\}$ and
\[
    q(\lambda)
    =\widehat\gamma'K\widehat\gamma
     -\widehat\gamma'K(K+\lambda W_n)^{-1}K\widehat\gamma.
\]
When $\widehat\gamma\ne0$,
\[
    q'(\lambda)=\gamma(\lambda)'W_n\gamma(\lambda)>0,
\]
so $q$ is strictly increasing and
\[
    \sup_{\lambda\ge0}q(\lambda)
    =\widehat\gamma'K\widehat\gamma.
\]
The supremum is approached only as $\lambda\to\infty$. If $\widehat\gamma=0$, both the primal and dual values equal zero for every $\lambda\ge0$.
\end{proof}

\begin{proof}[\textbf{Proof of Proposition \ref{prop:spectral_par}}]
    With $C=W_n^{1/2}$, write $u=C\gamma$ and $\widehat u=C\widehat\gamma$. Then
    \[
         N_n(b;\gamma)=(\widehat u-u)'\tilde K(\widehat u-u),
        \qquad
        \gamma'W_n\gamma=\|u\|_2^2,
    \]
    where $\tilde K=C^{-1}KC^{-1}\succ0$.

    If $g=0$, the feasible set in $u$-coordinates is $\{0\}$. Hence $u^*=0$ and $\gamma^*=0$, and
    \[
        N_n(b;\gamma^*)=\widehat u'\tilde K\widehat u.
    \]
    Using $\tilde K=Q\Lambda Q'$ and $\tilde u=Q'\widehat u$ gives $\widehat u'\tilde K\widehat u=\sum_{j=1}^r\Lambda_j\tilde u_j^2$. If $\widehat u\ne0$, $u(\lambda)=(\tilde K+
    \lambda I)^{-1}\tilde K\widehat u$ is nonzero for every finite $\lambda$. Thus, the secular equation with $g=0$ has no finite root.

    Suppose $g>0$. The KKT stationarity condition in $u$-coordinates is
    \[
        2\tilde K(u-\widehat u)+2\lambda u=0,
        \qquad
        u(\lambda)=(\tilde K+
        \lambda I)^{-1}\tilde K\widehat u .
    \]
    Diagonalizing $\tilde K=Q\Lambda Q'$ yields
    \[
        u(\lambda)=Q\,\mathrm{diag}\!\left(\frac{\Lambda_j}{\Lambda_j+\lambda}\right)Q'\widehat u,
    \]
    and therefore
    \[
        \phi(\lambda)=\|u(\lambda)\|_2^2-g^2
        =\sum_{j=1}^r\left(\frac{\Lambda_j}{\Lambda_j+\lambda}\right)^2\tilde u_j^2-g^2 .
    \]
    Differentiation gives
    \[
        \phi'(\lambda)
        =-2\sum_{j=1}^r \frac{\Lambda_j^2}{(\Lambda_j+\lambda)^3}\tilde u_j^2\le0,
    \]
    with strict inequality when $\widehat u\ne0$.

    If $\|\widehat u\|_2\le g$, the unconstrained minimizer is feasible. Thus, $u^*=\widehat u$, $\gamma^*=\widehat\gamma$, $\lambda^*=0$, and $ N_n(b;\gamma^*)=0$. If $\|\widehat u\|_2>g$, then $\phi(0)>0$ and $\lim_{\lambda\to\infty}\phi(\lambda)=-g^2<0$. Since $\phi$ is continuous and strictly decreasing in this case, there is a unique finite root $\lambda^*>0$, and $u^*=u(\lambda^*)$.

    Finally,
    \[
        \widehat u-u(\lambda)=Q\,\mathrm{diag}\!\left(\frac{\lambda}{\Lambda_j+\lambda}\right)Q'\widehat u.
    \]
    Hence, for finite $\lambda^*$,
    \[
        N_n(b;\gamma^*)=(\widehat u-u(\lambda^*))'\tilde K(\widehat u-u(\lambda^*))
        =\sum_{j=1}^r\Lambda_j\left(\frac{\lambda^*}{\Lambda_j+\lambda^*}\right)^2\tilde u_j^2.
    \]
    Furthermore,
    \[
        \lim_{\lambda \rightarrow \infty}\sum_{j=1}^r\Lambda_j\left(\frac{\lambda}{\Lambda_j+\lambda}\right)^2\tilde u_j^2 = \sum_{j=1}^r \Lambda_j \tilde u^2_j = \widehat \gamma' K \widehat \gamma
    \]
    coincides with the $g=0$ case.
\end{proof}

\begin{proof}[\textbf{Proof of Proposition \ref{prop:par_closed}}]
Let $u=W_n^{1/2}\gamma$ and $\widehat u=W_n^{1/2}\widehat\gamma$. Since $K=nW_n$,
\[
    (\widehat\gamma-\gamma)'K(\widehat\gamma-\gamma)
    =n\|\widehat u-u\|_2^2,
    \qquad
    \gamma'W_n\gamma=\|u\|_2^2.
\]
The profiled numerator is therefore the Euclidean projection problem
\[
    \min_{u} n\|\widehat u-u\|_2^2
    \quad\text{subject to}\quad
    \|u\|_2\le g.
\]
Projection onto the ball gives
\[
    \gamma^\ast=
    \begin{cases}
        \widehat\gamma, & \|\widehat\gamma\|_{W_n}\le g, \\
        \dfrac{g}{\|\widehat\gamma\|_{W_n}}\widehat\gamma,
          & \|\widehat\gamma\|_{W_n}>g.
    \end{cases}
\]
Hence
\[
    pAR\{b;\Gamma_n(g,W_n)\}
    =\frac{n(\|\widehat\gamma\|_{W_n}-g)_+^2}{\widehat\sigma_n^2(b)}
    =\frac{(\|\widehat\gamma\|_K-\sqrt n\,g)_+^2}{\widehat\sigma_n^2(b)}.
\]
The exact-exclusion benchmark satisfies
\[
    AR(b;0)=\frac{\|\widehat\gamma\|_K^2}{\widehat\sigma_n^2(b)}.
\]
Substituting $\|\widehat\gamma\|_K=\widehat\sigma_n(b)\sqrt{AR(b;0)}$ yields
\[
    pAR\{b;\Gamma_n(g,W_n)\}
    =\left(
        \sqrt{AR(b;0)}
        -\frac{\sqrt n\,g}{\widehat\sigma_n(b)}
      \right)_+^2.
\]
\end{proof}


\newpage
\bibliographystyle{unsrtnat}
\bibliography{pAR_references}
\end{document}